\documentclass[paper=letter, fontsize=12pt]{article} 
\RequirePackage[OT1]{fontenc}
\RequirePackage{amsthm,amsmath,amsfonts}
\usepackage{longtable}
\usepackage{graphicx}
\usepackage{latexsym}
\usepackage{multirow}
\RequirePackage[numbers]{natbib}
\RequirePackage[colorlinks,citecolor=blue,urlcolor=blue]{hyperref}
\usepackage{pst-all}
\usepackage{epsfig}
\usepackage{subcaption}

\allowdisplaybreaks[1]

\usepackage{times}
\usepackage{keyval}
\usepackage{calc}

\newlength{\oldparindent}
\numberwithin{equation}{section}
\theoremstyle{plain}
\newtheorem{thm}{Theorem}[section]
\newtheorem{assump}{Assumption}[section]
\newtheorem{rmk}{Remark}[section]
\newtheorem{lem}{Lemma}[section]
\newtheorem{cor}{Corollary}[section]
\newtheorem{algo}{Algorithm}[section]

\newcommand{\CEsp}[3]{ \textbf{E}^{#1}_{#3} \left[  #2  \right] }

\newcommand{\norm}[1]{\left| #1 \right|}
\newcommand{\normf}[2]{\| #1 \|_{#2} }

\newcommand{\Done}[2]{  {\frac{ \partial  #1}{\partial #2}}  }

\usepackage{fancyhdr}
\begin{document}

\rhead{May 19, 2022}
\lhead{P.\ Oyono~Ngou \& C.\ Hyndman  }
\chead{\textit{\quad \quad Fourier interpolation method for FBSDEs}}

\title{A Fourier interpolation method for numerical solution of FBSDEs: Global convergence, stability, and higher order discretizations \footnote{A previous version of this paper was titled ``Global convergence and stability of a convolution method for numerical solution of BSDEs'' \textit{arXiv:1410.8595v1}}}

\author{
Polynice OYONO~NGOU\footnotemark[3] \ and \ Cody HYNDMAN\footnote{Corresponding Author: email: cody.hyndman@concordia.ca}\ \footnote{ 
Department of Mathematics and Statistics, 
Concordia University, 
1455 Boulevard de Maisonneuve Ouest,
Montr\'eal, Qu\'ebec,
Canada H3G 1M8.
} 
\ 
}

\date{May 19, 2022}

\maketitle

\abstract{
The convolution method for the numerical solution
of forward-backward stochastic differential equations (FBSDEs), introduced in
\cite{hyndmanoyonongou:2013}, uses a uniform
space grid. %
In this paper we utilize a tree-like spatial discretization that approximates the BSDE on the tree, so that no spatial interpolation procedure is necessary.  In addition to suppressing extrapolation error, leading to a globally convergent numerical solution for the FBSDE, we provide explicit convergence rates.  On this alternative grid the conditional expectations involved in
the time discretization of the BSDE are computed using Fourier analysis and
the fast Fourier transform (FFT) algorithm.
The method is then extended to higher-order time discretizations of FBSDEs. Numerical results
demonstrating convergence are presented using a commodity price model, incorporating seasonality, and forward prices.
}

\vspace{5mm}


\vspace{5mm}

\section{Introduction\label{sec:Intro}}

A variety of numerical methods for backward stochastic differential equations (BSDEs) and forward-backward 
stochastic differential equations (FBSDEs) have been developed recently. Different applications call for innovative techniques for
the efficient resolution of these systems. In finance and economics
BSDEs are used for option pricing and hedging \cite{elkaouri:1996}, reflected BSDEs are used for modeling
American options \cite{elkarouietal:1997}, and
quadratic BSDEs play an important role in continuous-time recursive
utility \cite{duffieepstein:1992} and other utility maximization problems.  

Following the establishment of the well-posedness of BSDEs by \cite{pardouxpeng:1992} the
first numerical procedures to emerge were partial differential equation (PDE) based methods such as the finite difference approach of \cite{douglas:1996}. The PDE method is mainly devoted to coupled problems as is the spectral method of \cite{mashenzhao:2008}.  More recent numerical methods based on machine learning \cite{beck, weinan20, weinan17, weinan19, han20} have been applied to the numerical solution of BSDEs which, given the connections between the theory of PDE and their represeantations as BSDEs, in turn provides solutions to high dimensional PDEs.
Spatial discretization methods were initiated by \cite{chevance:1997}
with a quantization approach to conditional expectations.  However, it is
only since \cite{zhang:2004} and  \cite{bouchardtouzi:2004}
that a sound time discretization of (decoupled) FBSDEs is available.
The quantization approach was then used in the multidimensional framework
of  \cite{ballypages:2003,ballypagesp:2005} and the coupled FBSDE case of  \cite{delaruem:2006}. The theoretical basis for a multinomial approach for BSDEs was  introduced by \cite{briand:2001} and \cite{maprottermt:2002} and followed in practice by \cite{pengxu:2011}.
Monte Carlo methods are the most prolific approach for numerical
solutions of (F)BSDEs. They include the backward scheme of \cite{zhang:2004},
the Malliavin approach of  \cite{bouchardtouzi:2004}
and  \cite{crisanmt:2010}, the least-square regression
approach of  \cite{gobet:2005} and the iterative schemes
 \cite{benderdenk:2007} and \cite{benderzhang:2008}.

Additional approachs to numerical solution of BSDEs include,
include the cubature method  \cite{crisanm:2012,crisanmol:2013}, Fourier-cosine expansions \cite{RuijOosl:2013,oosterlee:2016,Ruijter20161},  and the convolution method  \cite{hyndmanoyonongou:2013}.
In introducing the convolution method, \cite{hyndmanoyonongou:2013} developed a local discretization
error which includes an extrapolation error. The extrapolation error component is
exclusively produced by the fast Fourier transform (FFT) algorithm
and the underlying trigonometric interpolation used to compute conditional
expectations. To improve the performance of the convolution method it is desirable to eliminate the extrapolation
error, and improve the error bound, with an alternative implementation of the FFT
algorithm.

In this paper we propose an alternative space grid for the convolution method,
instead of the rectangular grid in \cite{hyndmanoyonongou:2013}, which eliminates the extrapolation error, leads to a globally convergent numerical solution for the (F)BSDE, and provides explicit convergence rates.  We also apply the numerical method to the Runge-Kutta schemes for FBSDEs proposed by  \cite{Chasscrisan:2013}.  The tree-like nature of the alternative grid avoids extrapolations
and leads to a global error bound for the BSDE approximate solutions. 
Further, the implementation of the convolution method originally presented in  \cite{hyndmanoyonongou:2013}  is simplified by using an alternative parametric transformation to enforce the necessary periodic boundary conditions.

The paper is organized as follows.  Section \ref{sec:AltD} reviews the explicit Euler time discretization of BSDEs, recalls the  convolution method of  \cite{hyndmanoyonongou:2013} for computing the necessary conditional expectations,  gives a description of an alternative spatial discretization, and provides a generic implementation of the convolution method on this grid using the discrete Fourier transform.  The section ends with a global error analysis.  Section \ref{sec:Ext} extends the Fourier interpolation method to higher order time discretizations of FBSDEs and includes the related global error analysis. Finally, Section \ref{sec:NumResults}
presents  a numerical implementation,  in the context of a commodity price model, that illustrate the theoretical results. Section \ref{sec:Conclusion} concludes.

\section{The Fourier interpolation method\label{sec:AltD}}

In this section, we introduce an alternative grid that remove extrapolation error of \cite{hyndmanoyonongou:2013}, after a quick presentation of the Euler scheme for FBSDEs. Section \ref{sec:ErrorA} presents a global error analysis of the Fourier interpolation method on this alternative grid and under the Euler scheme.

\subsection{Time Discretization}

Let $\left( \Omega , {\bf P}, \mathcal{F}, \{ \mathcal{F}_t\}_{t \in [0,T]} \right)$
be a complete filtered probability space generated by a $d-$dimensional
Wiener process $W$. We seek a numerical solution to the FBSDE 

\begin{equation} 
\begin{cases} dX_t = a(t,X_t)dt + \sigma(t,X_t)dW_t \\
-dY_t = f(t,X_t,Y_t,Z_t)dt - Z^*_tdW_t \\ 
X_0 = x_0 , \ Y_T = \xi  \end{cases} 
\label{eq:bsde} 
\end{equation}

The forward drift $a: [0,T] \times \mathbb{R}^d \rightarrow \mathbb{R}^d$,
the forward volatility $\sigma: [0,T]\times\mathbb{R}^d \rightarrow \mathbb{R}^{d \times d}$,
the driver $f:[0,T] \times \mathbb{R}^d\times\mathbb{R}\times \mathbb{R}^d \rightarrow \mathbb{R}$
are deterministic functions. The initial condition $x_{0}\in\mathbb{R}^{d}$
and the terminal condition take the Markovian form $\xi=g(X_{T})$
where $g :\mathbb{R}^d \rightarrow \mathbb{R}$. The FBSDE coefficients
satisfy Assumption \ref{assump:coef} so that existence and uniqueness
of the FBSDE solution $(X,Y,Z)$ is assured.

\begin{assump}\label{assump:coef} There exist positive constants $K_1$, $K_2$ $K_3$, and $K_4$ such that the coefficients of the FBSDE (\ref{eq:bsde}) satisfy
\begin{align} 
\norm{a(t,x_1) - a(t,x_2)} + \normf{\sigma(t,x_1) - \sigma(t,x_2)}{2}  & \leq   K_1 \norm{x_1 - x_2}\\
\norm{a(t,x)} + \normf{\sigma(t,x)}{2}  & \leq  K_2 \\
\norm{ f(t,x_1,y,z) - f(t,x_2,y,z)}  &\leq  K_1 \norm{x_1 - x_2} \\
\norm{ f(t,x,y_1,z_1) - f(t,x,y_2,z_2)}  &\leq  K_1 \left(\norm{y_1 - y_2} + \norm{z_1 - z_2} \right)\\
\norm{ f(t,x,y,z) }  &\leq  K_3 (1 + \norm{x} +\norm{y} + \norm{z})
\end{align} 
for any $t\in[0,T]$, $x,x_1,x_2 \in \mathbb{R}^d$, $y, y_1,y_2 \in \mathbb{R}$, $z,z_1,z_2 \in \mathbb{R}^d$. 

Moreover $\sigma^2 := \sigma \sigma^*$ is (uniformly) invertible, continuous and bounded
\begin{align}  \normf{ (\sigma^2(t,x))^{-1} }{2}  \leq K_4 \label{eq:ssK}
\end{align} 
for any $t\in[0,T]$, $x \in \mathbb{R}^d$.

In addition, the terminal value is square integrable \begin{equation} \normf{\xi}{L^2}^2 := \CEsp{}{ \norm{g(X_T)}^2 }{} < \infty. \end{equation}
\end{assump} 

\begin{rmk}
  Assumption~\ref{assump:coef} makes no explicit assumption on $g$ except square integrability.  More restrictive assumptions on $g$, namely that $g$ is twice continuously differentiable, shall be specified in the main results of Section 2 and 3 which provide explicit rates of convergence for the Fourier interpolation algorithms.  However, similar to \cite{crisanm:2012}, see also \cite{10.1214/EJP.v20-3022}, it should be possible to consider $g$ non-differentiable using a mollification argument on the terminal condition.\footnote{This approach was suggested by an anonymous referee of an earlier version of this paper.}  However, we shall not follow this approach in this paper so as to keep the focus on the main contributions which are the overall approach, implementation of the convolution method on the tree-like grid, developing the Runge-Kutta discretization schemes, and explicit convergence rates.
\end{rmk}

The time discretization of the FBSDE (\ref{eq:bsde}) on the time partition
$\pi=\{t_{0}=0<t_{1}<\ldots<t_{n}=T\}$ consists of the explicit
Euler scheme given by 
\begin{equation}
  \begin{cases} {X}^{\pi}_0 = x_0 \\ 
{X}^{\pi}_{t_{i+1}} =  {X}^{\pi}_{t_i} + a(t_i, {X}^{\pi}_{t_i}) \Delta_i + \sigma(t_i, {X}^{\pi}_{t_i}) \Delta W_i \\ 
{Z}^{\pi}_{t_n} = 0 , \ {Y}^{\pi}_{t_n} = \xi^{\pi} \\
{Z}^{\pi}_{t_i} = \frac{1}{\Delta_i} \CEsp{}{ {Y}^{\pi}_{t_{i+1}} \Delta W_i | \mathcal{F}_{t_i} }{} \\  {Y}^{\pi}_{t_i}=\CEsp{}{{Y}^{\pi}_{t_{i+1}}|\mathcal{F}_{t_i}}{}+f({t_i},{X}^{\pi}_{t_i},\CEsp{}{{Y}^{\pi}_{t_{i+1}}|\mathcal{F}_{t_i}}{},{Z}^{\pi}_{t_i})\Delta_i
  \end{cases} \label{eq:expdisc} 
\end{equation} 
with
$\Delta_{i}=t_{i+1}-t_{i}$ and $\Delta W_{i}=W_{t_{i+1}}-W_{t_{i}}$.
Under an additional Lipschitz condition on the function $g$ we have, from \cite{zhang:2004} and \cite{bouchardtouzi:2004},
that the quadratic discretization error
\begin{equation}
  E^2_{\pi} := \max_{0 \leq i <n} \CEsp{}{ \sup_{t \in [t_{i},t_{i+1}]} \norm{Y_t - {Y}^{\pi}_{t_i} }^2  }{} + \sum_{i=0}^{n-1} \CEsp{}{  \int_{t_i}^{t_{i+1}} \norm{Z_s - Z^{\pi}_{t_i}}^2 ds}{}
\end{equation}
is of first-order in time, i.e
\begin{equation} E^2_{\pi} = \mathcal{O}(\norm{\pi}) \label{eq:CEuler}
\end{equation}
where 
\begin{equation} \norm{ \pi } = \max_i \Delta_i. \end{equation}  Following \cite{hyndmanoyonongou:2013} and \cite{poly:PhD}, the approximate solution $u_{i}$
and the approximate gradient $\dot{u}_{i}$ at time node $t_{i}$,
$i=0,1,\ldots,n-1$, are given by
\begin{align} u_i(x)  & =   \tilde{u}_i(x) + \Delta_i f(t_i, x, \tilde{u}_i(x) , \sigma(t_{i},x)\dot{u}_i(x)) \label{eq:uf}  \\
\sigma(t_i, x)\dot{ u}_i (x) & =  \frac{1}{\Delta_i} \CEsp{ }{ Y^{\pi}_{t_{i+1}} \sigma(t_i, {X}^{\pi}_{t_i})\Delta {W}_i | {X}^{\pi}_{t_i} = x}{} \nonumber\\
& =  \frac{1}{\Delta_i} \int_{\mathbb{R}^d}^{} (y-\Delta_i a(t_i,x)) u_{i+1}(x + y) h_i(y|x) dy  \label{eq:guf1}
\end{align} 
where the intermediate solution $\tilde{u}_i$ is given by
\begin{align} \tilde{u}_i(x) &=  \CEsp{}{Y^{\pi}_{t_{i+1}} | {X}^{\pi}_{t_i} = x}{} \nonumber\\
&=  \int_{\mathbb{R}^d}^{} u_{i+1}(y) h_i(y|x) dy, \label{eq:iuf1} \end{align}
$u_{n}=g$ and $h_i$ is a Gaussian density \begin{equation} h_i(y|x) = ( 2\pi  )^{-\frac{d}{2}} \normf{{\Delta_i}\sigma^2(t_i,x) }{2}^{-\frac{1}{2}} \exp\left(-\frac{1}{2 \Delta_i }  {\bf y}^* (\sigma^2(t_i,x))^{-1} {\bf y} \right),
\end{equation}
and where ${\bf y}=y-\Delta_{i}a(t_{i},x)$ with characteristic function
\begin{equation}
  \phi_i(\nu,x) = \exp \left( {\bf i} \Delta_i \nu^* a(t_i,x)  - \frac{1}{2}\Delta_i \nu^*\sigma^2(t_i,x) \nu  \right)  .
\end{equation}

Let $\mathfrak{F}$ and $\mathfrak{F}^{-1}$ denote the Fourier transform
operator and the inverse Fourier transform operator respectively,
\begin{align}
  \mathfrak{F}[\theta](\nu) &= \int_{\mathbb{R}^d} e^{- {\bf i} x^* \nu } \theta(x) dx \\
  \mathfrak{F}^{-1}[\theta](x) &= \frac{1}{(2\pi)^d} \int_{\mathbb{R}^d} e^{ {\bf i} \nu^* x } \theta(\nu) d\nu.
\end{align}
Using the relationships between the characteristic function and the density function leads to the representation
\begin{align} %
\tilde{u}_i(x) &= \mathfrak{F}^{-1}[ \mathfrak{F}[u_{i+1}](\nu) \phi_{i}(\nu,x) ](x) \label{eq:guf2}\\
\dot{u}_i(x) &= \sigma^*(t_i,x) \mathfrak{F}^{-1}[ \mathfrak{F}[u_{i+1}](\nu) {\bf i} \nu \phi_{i}(\nu,x) ](x) \label{eq:iuf2}
\end{align}
for equations (\ref{eq:guf1}) and (\ref{eq:iuf1}) under integrability
condition on the approximate solution $u_{i+1}$. In the sequel, we restrict the analysis to the one-dimensional case with $d=1$. %

\subsection{Space discretization}

The space discretization is performed on a tree-like grid using three
parameters: the increment length $l>0$, the even number
$N\in\mathbb{N}^{*}$ of space steps on the increment length, and the
initial number $N_{0}\in\mathbb{N}$ of increment intervals. Hence,
the space step is constant and uniform on the grid \begin{equation} \Delta x = \frac{l}{N}.\end{equation}
At the time node $t_{i}$, $i=0,1,\ldots,n$, the space domain is restricted
on an interval of length $N_{i}l$ centred at $x_0$ and discretized
uniformly with $N_{i}N$ space steps where \begin{equation} N_i = N_0 + i \end{equation}
giving the nodes
\begin{equation} x_{ik} = x_0 -\frac{N_il}{2} + k \Delta x  , \quad k=0,1,\ldots,N_i N.
\end{equation}
In particular, the relation
\begin{equation}
  x_{ik} = x_{i+1,k+\frac{N}{2}} , \quad k=0,1,\ldots,N_iN. \label{eq:intergrid}
\end{equation}
holds since the restricted interval at each time node is obtained by evenly
increasing the previous one with an interval of length $l$. If $N_{0}=0$
then the space grid at mesh time $t_{0}$ is compose by the single
point \begin{equation} x_{00} = x_0.\end{equation} Figure \ref{fig:grid}
gives examples of alternative grids. 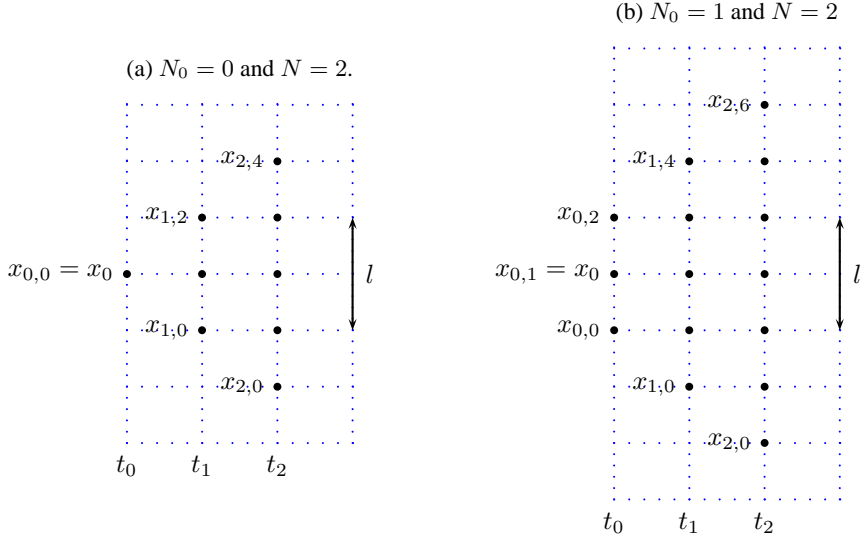
\begin{figure}

\centering
\caption{Examples of alternative grids}
\vspace{0.25cm}
\psset{showpoints = false, dotsize =3pt, griddots = 5, gridlabels = 0pt, xunit = 1, yunit = 0.75, gridcolor = blue, subgridcolor=lightgray}

\mbox{
\begin{subfigure}[c]{0.3 \textwidth} \centering \caption{$N_0=0$ and $N=2$.}
\begin{pspicture}(0,-3)(3,3) \psgrid[subgriddiv = 1, subgriddots = 5] 
\psdots(0,0)(1,0)(1,1)(1,-1)(2,0)(2,-2)(2,-1)(2,2)(2,1)
\uput[d](0,-3){$t_0$} 
\uput[d](1,-3){$t_1$} 
\uput[d](2,-3){$t_2$} 
\uput[r](3,0){$l$} 
\psline{<->}(3,-1)(3,1)
\uput[l](0,0){ $x_{ 0,0}=x_0$} 
\uput[l](1,-1){ $x_{1,0}$} 
\uput[l](1,1){ $x_{1,2}$} 
\uput[l](2,-2){ $x_{2,0}$} 
\uput[l](2,2){ $x_{2,4}$}
\end{pspicture}
\end{subfigure}

\hspace{1.5cm}

\begin{subfigure}[c]{0.3 \textwidth} \centering \caption{$N_0=1$ and $N=2$}
\begin{pspicture}(0,-4)(3,4) \psgrid[subgriddiv = 1, subgriddots = 5]
\psdots(0,0)(0,1)(0,-1)(1,0)(1,1)(1,-1)(1,2)(1,-2)(2,0)(2,-2)(2,-1)(2,2)(2,1)(2,-3)(2,3)
\uput[d](0,-4){$t_0$} 
\uput[d](1,-4){$t_1$} 
\uput[d](2,-4){$t_2$} 
\uput[r](3,0){$l$} 
\psline{<->}(3,-1)(3,1)
\uput[l](0,-1){ $x_{0,0}$} 
\uput[l](0,0){ $x_{0,1} = x_0$}
\uput[l](0, 1){ $x_{0,2}$} 
\uput[l](1,-2){ $x_{1,0}$} 
\uput[l](1,2){ $x_{1,4}$} 
\uput[l](2,-3){ $x_{2,0}$} 
\uput[l](2,3){ $x_{2,6}$}
\end{pspicture}
\end{subfigure}
}
\label{fig:grid}
\vspace{0.25cm}
\end{figure}

The convolution relations of equations (\ref{eq:guf2}) and (\ref{eq:iuf2})
call for a discretization of the Fourier space as well. At each mesh
time $t_{i}$, $i=1,2,\ldots,n$, the Fourier space is restricted on
an interval of length $L$ centred at zero $(0)$ and discretized
with $N_{i}N$ space steps. The equidistant nodes are thus of the
form \begin{equation} \nu_{ik} = -\frac{L}{2} + k \Delta \nu_i \text{, } k = 0,1,\ldots,N_iN \end{equation}
where $\Delta \nu_i=\frac{L}{N_iN}$. The {\it Nyquist~relation}\footnote{The minimum sampling rate to avoid aliasing.} holds whenever
$L$ is such that \begin{equation} Ll = 2\pi N.\end{equation}

\subsection{Implementation} \label{sec:implement}

In order to compute numerical approximations of equations (\ref{eq:guf2}) and (\ref{eq:iuf2})
at time node $t_{i}$, $i=0,1,\ldots,n-1$, we introduce the generic
functions $\theta_{i}:\mathbb{R}\rightarrow\mathbb{R}$, $\psi:\mathbb{R}^2\rightarrow\mathbb{C}$
and $\theta_{i+1}:\mathbb{R}\rightarrow\mathbb{R}$ such that \begin{equation} \theta_i (x) = \mathfrak{F}^{-1} \left[ \mathfrak{F}[\theta_{i+1}](\nu)  \psi( \nu, x )\right](x). \label{eq:Pint2}
\end{equation}We assume that the function $\theta_{i+1}$ satisfies the periodicity boundary
value equalities of Assumption \ref{assump:val}.
\begin{assump} \label{assump:val}
 The generic function $\theta_{i+1}$ satisfies \begin{align} \theta_{i+1}\left( x_{i+1,0}  \right) &= \theta_{i+1}\left( x_{i+1,N N_{i+1}}  \right) \label{eq:bound1} \\
\Done{ \theta_{i+1}}{x}\left( x_{i+1,0}  \right)  &=  \Done{\theta_{i+1}}{x}\left( x_{i+1,N N_{i+1}}  \right). \label{eq:bound2}
\end{align} \end{assump}
Hence, the Fourier integral 
\begin{equation} 
   \mathfrak{F}[\theta_{i+1}](\nu) = \int_{-\infty}^{\infty} e^{ -{\bf i}\nu x} \theta_{i+1}(x)dx 
\end{equation}
is restricted on the interval $[x_0 -\frac{N_{i+1}l}{2},x_0 + \frac{N_{i+1}l}{2}]= [x_{i+1,0},  x_{i+1,N N_{i+1}} ]$
and discretized using the grid points $\{x_{i+1,k}\}_{k=0}^{N_{i+1}N}$
with a quadrature rule with weights $\{w_{k}\}_{k=0}^{N_{i+1}N}$.
As to the inverse Fourier integral of equation (\ref{eq:Pint2}) we
restrict it on the interval $[-\frac{L}{2},\frac{L}{2}]$ and discretize
it with lower Riemann sums at the Fourier space grid point $\{\nu_{i+1,k}\}_{k=0}^{N_{i+1}N}$. 

Let $\mathcal{D}$ and $\mathcal{D}^{-1}$ denote the discrete Fourier
transform and the inverse discrete Fourier transform respectively
\begin{align} \mathfrak{D}[ \{x\}_{i=0}^{N-1} ]_k &= \frac{1}{N} \sum_{j=0}^{N-1} e^{-{\bf i} j k \frac{2\pi}{N} } x_j \\
\mathfrak{D}^{-1}[ \{x\}_{i=0}^{N-1} ]_k &= \sum_{j=0}^{N-1} e^{{\bf i} j k \frac{2\pi}{N} } {x}_j.\end{align} Then the discretization procedure leads to the approximation \begin{equation} \theta_i (x_{i+1,k})  \approx  (-1)^{k} \mathfrak{D}^{-1}\left[ \left\{ \psi(\nu_{i+1,j}, x_{i+1, k}) \mathbb{D}[ \theta_{i+1} ]_j     \right\}_{j=0}^{N_{i+1}N-1} \right]_{k} \label{eq:numconv2pp} \end{equation}
where \begin{equation} \mathbb{D}[ \theta_{i+1} ]_j = \mathfrak{D}\left[ \{ (-1)^s \tilde{w}_s \theta_{i+1} (x_{i+1,s})\}_{s=0}^{N_{i+1}N-1} \right]_j \end{equation}
and the weights $\{\tilde{w}\}_{k=0}^{N_{i+1}N-1}$ are given by \begin{equation} \tilde{w}_{k}=w_{k}+w_{N_{i+1}N}\delta_{k,0}. \end{equation}
where $\delta$ stands for the Kronecker delta. The relation of equation (\ref{eq:intergrid}) allows us to write 
\begin{equation}
  \theta_i (x_{ik})  \approx  (-1)^{k + \frac{N}{2}} \mathfrak{D}^{-1}\left[ \left\{ \psi(\nu_{i+1,j}, x_{i k}) \mathbb{D}[ \theta_{i+1} ]_j     \right\}_{j=0}^{N_{i+1}N-1} \right]_{k+\frac{N}{2}} \label{eq:numconv2p}
\end{equation}
for $k=0,1,\ldots,N_{i}N$.

In equation (\ref{eq:numconv2p}), the generic function $\psi$ depends
on the space node $x_{ik}$. If the relation generalizes for all space
nodes $x_{ik}$, $k=0,1,\ldots,N_{i}N$, the function values $\theta_{i}(x_{ik})$,
$k=0,1,\ldots,N_{i}N$, can not be computed with a single direct FFT
procedure. Instead, a separate FFT procedure using the values of the
generic function $\psi$ at $x_{ik}$ is needed to compute the function
value $\theta_{i}(x_{ik})$. Nonetheless, the vector-matrix representation
of the FFT procedure in equation (\ref{eq:numconv2p}) allows the computation
of all function values $\theta_{i}(x_{ik})$ with a matrix multiplication. In the vector-matrix representation, equation (\ref{eq:numconv2p})
is
\begin{equation*}
  \theta_i (x_{ik}) = (-1)^{k+\frac{N}{2}} \hat{F}_{k+\frac{N}{2}} \Psi(x_{ik})\mathbb{D}[ \theta_{i+1} ]
\end{equation*}
where
$\hat{F}_{k+\frac{N}{2}}$ is the $(k+\frac{N}{2})$th row of the
$N_{i+1}N$ dimension inverse FFT matrix $\hat{F}$ and $\Psi(x_{ik})$
is the $N_{i+1}N$ dimension diagonal matrix built with the values
$\{\psi(\nu_{i+1,j},x_{ik})\}_{j=0}^{N_{i+1}N-1}$. Let ${\Theta}^{(i)}$
be the $N_{i}N$ dimension vector of the function values $\theta_{i}(x_{ik})$
such that \begin{equation} {\Theta}^{(i)}_{1+k} =  \theta_{i}(x_{ik}) \end{equation}
for $k=0,1,\ldots,N_{i}N$. The matrix representation gives \begin{equation} {\Theta}^{(i)} = \hat{ \Psi }^{(i)} \mathbb{D}[ \theta_{i+1} ] \end{equation}
where $\hat{\Psi}^{(i)}$ is the $(N_{i}N+1)\times N_{i+1}N$ matrix
such that \begin{equation} \hat{ \Psi }^{(i)}_{1+k,1+j} = (-1)^{k + \frac{N}{2}}\bar{\omega}_i^{j(k+\frac{N}{2})} \psi(\nu_{i+1,j},x_{ik}) \end{equation}with
$\bar{\omega}_{i}=e^{{\bf i}2\pi(N_{i+1}N)^{-1}}$, $k=0,1,\ldots,N_{i}N$
and $j=0,1,\ldots,N_{i+1}N-1$.

The requirements of Assumption \ref{assump:val} can easily be satisfied.
Given a function $\eta:[a,b]\rightarrow\mathbb{R}$ and $\eta\in\mathcal{C}^{1}$,
if we consider the transformation \begin{align} \eta^{\alpha,\beta}(x) = \eta (x) + \alpha x^2 + \beta x \label{eq:trans2} \end{align}
then the parameters $\alpha$ and $\beta$ can be chosen such that
the transform function and its derivative have equal values at the boundaries
of any interval. The following lemma gives a method to select the
coefficients $\alpha$ and $\beta$ for the transform of equation
(\ref{eq:trans2}) such that Assumption~\ref{assump:val} holds in general.
\begin{lem}
\label{lem:paramvalue}Suppose the real function $\eta\in\mathcal{C}^{1}[a,b]$
is differentiable and let $\eta^{\alpha,\beta}$ be its transformed
function as defined in equation (\ref{eq:trans2}). Then \begin{align} \alpha &= \frac{\Done{\eta}{x}(a) - \Done{ \eta }{x}(b)}{2(b-a)} \text{ , } \label{eq:optalpha2} \\  
\beta &= \frac{{\eta}(a) - { \eta }(b)}{(b-a)} - \alpha (b+a)   \label{eq:optbeta2}
\end{align}solve the system of linear equations defined by \begin{equation} \begin{cases}  \eta^{\alpha,\beta}(a) = \eta^{\alpha,\beta}(b) \\
\Done{ \eta^{\alpha,\beta} }{x}(a)  = \Done{ \eta^{\alpha,\beta} }{x}(b). \end{cases} \label{eq:sysfft3} \end{equation}\end{lem} 
\begin{proof} The second equation of the system (\ref{eq:sysfft3}) gives equation (\ref{eq:optalpha2})
in a straightforward manner. Equation (\ref{eq:optbeta2}) is given
by the first equation of the system.
\end{proof}
Hence, the numerical discretization may be applied on the transformation
$u_{i+1}^{\alpha,\beta}$ at time node $t_{i}$ but a correction must
be performed so to recover the values of the intermediate solution
$\tilde{u}_{i}$ and the approximate gradient $\dot{u}_{i}$. The
next theorem gives the representation under the transform
of equation (\ref{eq:trans2}). 
\begin{thm}
\label{lem:altT}
Let $u^{\alpha, \beta}_{i+1}$ be the alternative
transform defined in equation (\ref{eq:trans2}) of the approximate
solution $u_{i+1}$. Then the intermediate solution $\tilde{u}_{i}$
and the approximate gradient $\dot{u}_{i}$ in equations (\ref{eq:guf2}) and (\ref{eq:iuf2})
satisfy
\begin{align} \tilde{u}_i(x) &= \mathfrak{F}^{-1}[ \mathfrak{F}[u^{\alpha,\beta}_{i+1}](\nu) \phi(\nu,x) ](x) -  \alpha [ (x + \Delta_i a(t_i,x))^2 + \Delta_i \sigma^2(t_i,x)] \nonumber \\ &
  - \beta (x + \Delta_i a(t_i,x))  \label{eq:isfbsde3} \\
\dot{ u}_i (x) &= \sigma(t_i,x) \mathfrak{F}^{-1}[ \mathfrak{F}[u^{\alpha,\beta}_{i+1}](\nu) {\bf i} \nu \phi(\nu,x) ](x) - \sigma(t_i,x) [ 2 \alpha(x + \Delta_i a(t_i,x))  + \beta] .\label{eq:agfbsde3} 
\end{align} \end{thm}
\begin{proof}
The proof follows the steps of Theorem 3.1 of
\cite{hyndmanoyonongou:2013} using the transformation introduced
in equation (\ref{eq:trans2}).
\end{proof}

Algorithm \ref{alg:conv} details the numerical procedure on the space grid and produces numerical solutions $\{ {u}_{ik} \}_{k=0}^{N_iN}$,
$\{ {\tilde{u}}_{ik} \}_{k=0}^{N_iN}$ and $\{ {\dot{u}}_{ik} \}_{k=0}^{N_iN}$
for the approximate solution $u_{i}$, the intermediate solution $\tilde{u}_{i}$ and
the approximate gradient $\dot{u}_{i}$ respectively, $i=0,1,\ldots,n-1$.
We next consider error estimates under the alternative  discretization.

\begin{algo}  \label{alg:conv} Fourier Interpolation Method on Alternative Grid
\begin{enumerate}
\item Discretize the restricted real space $[x_0-\frac{N_nl}{2},x_0 + \frac{N_nl}{2}]$ and the restricted Fourier space $[-\frac{L}{2},\frac{L}{2}]$ with $N_nN$ space steps so to have the real space nodes $\{ x_{nk}\}_{k=0}^{N_nN}$ and $\{ \nu_{nk}\}_{k=0}^{N_nN}$
\item Value $u_n(x_{nk}) = g(x_{nk})$
\item For any $i$ from $n-1$ to $0$
\begin{enumerate}
\item Compute $\alpha$ and $\beta$ defining the transform of equation (\ref{eq:trans2}), such that
\begin{equation} \theta_{i+1}  = \ u_{i+1} ^{\alpha,\beta} \end{equation} and $\theta_{i+1}$ satisfies the boundary conditions of equations (\ref{eq:bound1}) and (\ref{eq:bound2}).
\item Compute $\theta_i(x_{ik})$ through equation (\ref{eq:numconv2p}) for $k=0,1,\ldots,N_iN$ with \begin{equation} \psi(\nu, x) = \phi_i(\nu, x) \end{equation} and retrieve the values $\tilde{u}_{ik}$ as 
  \begin{align}
    \tilde{u}_{ik} &= \theta_i(x_{ik}) - \alpha [ (x_{ik} + \Delta_i a(t_i,x_{ik}))^2 + \Delta_i \sigma^2(t_i,x_{ik})] \nonumber \\
    &- \beta (x_{ik} + \Delta_i a(t_i,x_{ik})).
  \end{align}
\item Compute $\theta_i(x_{ik})$ through equation (\ref{eq:numconv2p}) for $k=0,1,\ldots,N_iN$ with \begin{equation} \psi(\nu, x) = {\bf i}\nu\sigma(\nu, x)\phi_i(\nu, x) \end{equation} and retrieve the values $\dot{u}_{ik}$ as 
  \begin{equation}
    \dot{u}_{ik} = \theta_i(x_{ik}) - \sigma(t_i,x_{ik}) [ 2 \alpha(x_{ik} + \Delta_i a(t_i,x_{ik}))  + \beta].  \end{equation}
\item Compute the values $u_{ik}$ as 
\begin{equation} u_{ik} = \tilde{u}_{ik} + \Delta_i f(t_i,x_{ik}, \tilde{u}_{ik}, \dot{u}_{ik}) \end{equation} for $k=0,1,\ldots,N_iN$ through equation (\ref{eq:uf}).
\item Update the real space grid with equation (\ref{eq:intergrid}) and the Fourier space grid by discretizing the interval $[-\frac{L}{2},\frac{L}{2}]$ with $N_iN$ space steps so to have the real space nodes $\{ x_{ik}\}_{k=0}^{N_iN}$ and $\{ \nu_{ik}\}_{k=0}^{N_iN}$.
\end{enumerate}
\end{enumerate}
\end{algo}

\subsection{Spatial discretization error analysis\label{sec:ErrorA}}

Let $\{ \mathbf{u}_{ik} \}_{k=0}^{N_iN}$, $\{ \mathbf{\tilde{u}}_{ik} \}_{k=0}^{N_iN}$
and $\{ \mathbf{\dot{u}}_{ik} \}_{k=0}^{N_iN}$ denote the numerical
solutions obtained from the convolution method at time node $t_i$
given the solution $u_{i+1}$ at time $t_{i+1}$. For the Fourier interpolation
method on the alternative grid, we defined the local discretization
error as \begin{equation} { E}_{ik} := \norm{u_i(x_k) - \mathbf{u}_{ik}} +  \norm{\dot{u}_i(x_k) - \mathbf{\dot{u}}_{ik}} \end{equation}
for $i=0,1,\ldots,n-1$ and $k=0,1,\ldots,N_{i}N$. 
\begin{thm}
\label{thm:ADLerr}Suppose that the driver is $f\in\mathcal{C}^{1,2}([0,T]\times\mathbb{R}^{3})$,
the terminal condition is $g\in\mathcal{C}^{2}(\mathbb{R})$, and Assumption \ref{assump:coef}
is satisfied. Then
the convolution method yields a local discretization error of the
form \begin{equation} { E}_{ik} = \mathcal{O}\left(\Delta x \right) + \mathcal{O}\left( e^{-K \norm{\Delta_i}^{-1} l^2} \right) \label{eq:ADerrbound} \end{equation}for
some constant $K>0$ on the alternative grid and under the trapezoidal
quadrature rule with weights \begin{equation*} w_j = 1 - \frac{1}{2}(\delta_{j,0} + \delta_{j,N_{i+1}N}). \end{equation*} 
\begin{proof}
We suppose the solution $u_{i+1}$ at time $t_{i+1}$ is known. The solution $u_{i+1} \in \mathcal{C}^2$ is twice differentiable  since $f \in \mathcal{C}^{1,2}$ and $g \in \mathcal{C}^2$. Also, $u_{i+1}$ is square integrable with respect to the Gaussian density. 

By Theorem \ref{lem:altT}, we limit ourselves to the case where
\begin{align*}
  u_{i+1} \left(x_0 -\frac{N_{i+1}l}{2} \right) &= u_{i+1} \left(x_0+ \frac{N_{i+1}l}{2} \right) \\
 \Done{u_{i+1}}{x}\left(x_0 -\frac{N_{i+1}l}{2} \right) &= \Done{u_{i+1}}{x}\left( x_0 + \frac{N_{i+1}l}{2} \right)
\end{align*}
so that the coefficients of the transform are $\alpha =\beta =0$. Let $T_i $ be the Fourier polynomial interpolating $u_{i+1}$ on 
$\left[x_0 -\frac{N_{i+1}l}{2},x_0 + \frac{N_{i+1}l}{2}\right]$. Then
\begin{align} 
T_i (x) & :=  \sum_{k=-\frac{N_{i+1} N}{2}}^{\frac{N_{i+1}N}{2}-1} d_j e^{ {\bf i} k \frac{2 \pi }{N_{i+1}l} x }    \label{eq:fint1} \\
&= u_{i+1}(x) + \mathcal{O}(\Delta x), \quad \forall x \in \left[x_0 -\frac{N_{i+1}l}{2},x_0 + \frac{N_{i+1}l}{2}\right]\label{eq:fint2}
\end{align} where \begin{equation} (-1)^{j - \frac{N_{i+1}N}{2} }d_{j-\frac{N_{i+1}N}{2} } = \mathbb{D}[u_{i+1}]_{j} \text{, } j = 0,1,\ldots,N_{1+i}N-1 \label{eq:fint3} \end{equation} when using the trapezoidal quadrature rule. We have that 
\begin{align*} 
\tilde{u}_{i}(x_{ik}) &= \int_{ \norm{y} \leq \frac{l}{2} } u_{i+1}(x_{ik} + y) h_i(y|x_{ik})dy + \int_{ \norm{y} > \frac{l}{2} } u_{i+1}( x_{ik}+y) h_i(y|x_{ik})dy 
\end{align*} 
where 
\begin{align*} 
\int_{ \norm{y} > \frac{l}{2} } u_{i+1}( x_{ik}+y) h_i(y|x_{ik})dy &= \mathcal{O}\left(e^{-K l^2}\right)
\end{align*} 
for some constant $K>0$ which is inversely proportional to $\Delta_i$ by Cauchy-Schwarz and Chernoff inequalities since the solution $u_{i+1}$ is square integrable. 
Hence 
\begin{align*} 
&\tilde{u}_{i}(x_{ik})  =  \int_{ \norm{y} \leq \frac{l}{2} } u_{i+1}(x_{ik} + y) h_i(y|x_{ik})dy +  \mathcal{O}\left(e^{-Kl^2}\right) \\ 
& =  \int_{ \norm{y} \leq \frac{l}{2} } T_{i}(x_{ik} + y) h_i(y|x_{ik})dy + \mathcal{O}\left( \Delta x \right) +   \mathcal{O}\left(e^{-Kl^2}\right) \tag{by equation \ref{eq:fint2}}\\
& =  \int_{ \mathbb{R} } T_{i}(x_{ik} + y) h_i(y|x_{ik}) dy + \mathcal{O}\left( \Delta x \right) +   \mathcal{O}\left(e^{-Kl^2}\right) \tag{by Chernoff inequality, since $T_i$ is bounded} \\
& =  \int_{\mathbb{R}} \sum_{j = -\frac{N_{i+1}N}{2}}^{\frac{N_{i+1}N}{2}-1} d_{j} e^{{\bf i} j \frac{2\pi}{N_{i+1}l} (x_{i,k}+y)} h_i(y|x_{ik}) dy  +   \mathcal{O}(\Delta x) + \mathcal{O}\left(e^{-Kl^2}\right) \\
& =  \sum_{j = -\frac{N_{i+1}N}{2}}^{\frac{N_{i+1}N}{2}-1} d_{j} e^{{\bf i} j \frac{2\pi}{N_{i+1}l} x_{i,k}} \phi_i \left(j \frac{2\pi}{N_{i+1}l}, x_{ik}\right) +   \mathcal{O}(\Delta x) + \mathcal{O}\left(e^{-Kl^2}\right) \\
& =  (-1)^{k+\frac{N}{2}} \sum_{j=0}^{N_{i+1}N-1} \phi( \nu_{i+1,j}, x_{ik} ) \mathbb{D}[u_{i+1}]_j  e^{{\bf i} \frac{2\pi}{N_{i+1}N} j(k+\frac{N}{2})} +  \mathcal{O}(\Delta x) \\ &+ \mathcal{O}\left(e^{-Kl^2}\right)   \tag{by equation \ref{eq:fint3} when using the trapezoidal quadrature rule} \\ 
& =  \mathbf{\tilde{u}}_{ik} + \mathcal{O}(\Delta x) + \mathcal{O}\left(e^{-Kl^2}\right). 
\end{align*} 

Similar techniques show that 
\begin{equation} 
\dot{u}_{i}(x_k) = \mathbf{\dot{u}}_{ik} + \mathcal{O}\left(\Delta x\right) + \mathcal{O}\left( e^{-Kl^2}\right) 
\end{equation} 
where $K>0$ is inversely proportional to $\Delta_i$. 
The Lipschitz property of the driver $f$ completes the proof. \end{proof}
\end{thm}
As expected, the alternative discretization improves the local error
bound by eliminating extrapolation errors in \cite{hyndmanoyonongou:2013}. The result of Theorem \ref{thm:ADLerr}
establishes the consistency of the convolution method with respect
to the approximate functions $u_{i}$ and gradients $\dot{u}_{i}$.
Furthermore, the absence of extrapolation errors in the local discretization
allows us to develop a tighter bound for the global discretization error. The
following corollary proves helpful when deriving the global discretization
error bound.

\begin{cor}
\label{cor:LDerr}Under the conditions of Theorem \ref{thm:ADLerr}, there is $C>0$ such that
\begin{equation} \sup_{i,k} { E}_{i,k} = \mathcal{O}(\Delta x ) + \mathcal{O}\left(e^{-C \norm{\pi}^{-1} l^{2}}\right).\end{equation}
\end{cor}
We define the global error as \begin{equation} E_{l,\Delta x} := \sup_{i,k} e_{ik} + \sup_{i,k}\dot{e}_{ik} \label{eq:globdiser} \end{equation}
where \begin{align} e_{ik} &= \norm{u_{n-i}(x_k) - {u}_{n-i,k}} \\
\dot{e}_{ik} &= \norm{\dot{u}_{n-i}(x_k) - {\dot{u}}_{n-i,k}} \end{align}for $i=1,\ldots,n$ with $e_{0,k}=\dot{e}_{0,k}=0$. The next theorem
describes the stability and convergence properties of the convolution
method.
\begin{thm}
\label{thm:GDR}Suppose the conditions of Theorem \ref{thm:ADLerr}
are satisfied. If the space discretization is such that \begin{equation} \sup_{ i }  \max\left(\frac{K_4^{\frac{1}{2}}\Delta x}{\sqrt{2\pi  \Delta_i}} , \frac{K_4 \Delta x}{ \pi \Delta_i} \right)  \leq 1    \label{eq:ccond-1} \end{equation}then
the Fourier interpolation method is stable and the global discretization
error $E_{l,\Delta x}$ satisfies \begin{equation} E_{l,\Delta x} = \mathcal{O}(\Delta x ) + \mathcal{O}\left(e^{-C \norm{\pi}^{-1} l^{2}}\right) \label{eq:GDR-1} \end{equation}
where $C>0$ and $K_{4}$ is the upper bound of equation (\ref{eq:ssK}).\end{thm}
\begin{proof}
Let's first notice that \begin{align}
e_{ik} & \leq { E}_{n-i,k} + \norm{ \mathbf{{u}}_{n-i,k} - {u}_{n-i,k}}  \nonumber\\
& \leq E_{n-i,k} + (1+\Delta_iK) \norm{ \mathbf{\tilde{u}}_{n-i,k} - \tilde{u}_{n-i,k}} + \hspace{1mm}  \Delta_i K \norm{ \mathbf{\dot{u}}_{n-i,k} - {\dot{u}}_{n-i,k}} \label{eq:2prf1}
\end{align}where $K>0$ is the Lipschitz constant of the driver $f$. Also, we
have that \begin{equation} \dot{e}_{ik} \leq {E}_{n-i,k} + \norm{ \mathbf{\dot{u}}_{n-i,k} - {\dot{u}}_{n-i,k}}. \label{eq:2prf12} \end{equation}Furthermore,
the construction of the Fourier interpolation method gives \begin{align} 
\norm{ \mathbf{\tilde{u}}_{i,k} - \tilde{u}_{i,k}} & \leq \norm{ \mathfrak{D}^{-1}\left[ \left\{ \phi(\nu_{i+1,j}, x_{ik}) \mathbb{D}[ u_{i+1} - u_{i+1,s} ]_j     \right\}_{j=0}^{N_{i+1}N-1} \right]_{k+\frac{N}{2}} }  \nonumber\\
& \leq \frac{1}{N_{i+1}N} \left( \sum_{j=0}^{N_{i+1}N-1} \norm{\phi(\nu_{i+1,j},x_{ik})} \right) \sup_k \norm{ u_{i+1}(x_{ik}) - u_{i+1,k} } 
  \tag{using the matrix-vector representation of DFTs}\nonumber\\
& \leq \frac{1}{N_{i+1}N} \left( \sum_{j=0}^{N_{i+1}N-1}\norm{ \phi(\nu_{i+1,j} , x_{ik})} \right)  \sup_k e_{n-i-1,k}  \nonumber\\ 
& \leq \frac{ (\Delta \nu_{i+1})^{-1} }{N_{i+1}N} \left( \int_{\mathbb{R}} \norm{ \phi(\nu, x_{ik}) } d\nu \right) \sup_k e_{n-i-1,k} \nonumber\\
& \leq \frac{ K_4^{\frac{1}{2}}\Delta x }{(2\pi \Delta_i)^{\frac{1}{2}} } \sup_k e_{n-i-1,k}. \label{eq:2prf2}
\end{align}
where the last inequality holds by Assumption \ref{assump:coef}.
Similarly, \begin{align} 
\norm{ \mathbf{\dot{u}}_{i,k} - \dot{u}_{i,k}} & \leq  \norm{ \mathfrak{D}^{-1}\left[ \left\{ \psi(\nu_{i+1,j}, x_{ik}) \mathbb{D}[ u_{i+1} - u_{i+1,s} ]_j     \right\}_{j=0}^{N_{i+1}N-1} \right]_{k+\frac{N}{2}} } \nonumber\\
& \leq \frac{1}{N_{i+1}N} \left( \sum_{j=0}^{N_{i+1}N-1} \norm{ {\bf i}\nu_{i+1,j} \phi(\nu_{i+1,j}, x_{ik}) } \right) \sup_k e_{n-i-1,k} 
  \tag{using the matrix-vector representation of DFTs}\nonumber\\
& \leq \frac{(\Delta \nu_{i+1})^{-1} }{N_{i+1}N} \left( \int_{\mathbb{R}} \norm{\nu \phi(\nu,x_{ik} ) }d \nu \right) \sup_k e_{n-i-1,k} \nonumber\\
& \leq \frac{ K_4 \Delta x}{\pi \Delta_i} \sup_k e_{n-i-1,k}. \label{eq:2prf3}
\end{align}

The inequalities of equations (\ref{eq:2prf1}), (\ref{eq:2prf2}) and (\ref{eq:2prf3})
lead to \begin{align} 
e_{i,k} & \leq C_0 { E}_{i,k} + \left({1+ 2\Delta_iK} \right) \max\left(\frac{K_4^{\frac{1}{2}} \Delta x}{\sqrt{2\pi \Delta_i}} , \frac{K_4\Delta x}{ \pi \Delta_i} \right)  \sup_k e_{i-1,k} \nonumber\\
& \leq C_0 \sup_{i,k}{ E}_{i,k} + \left({1+ 2\Delta_i K} \right)\max\left( \frac{K_4^{\frac{1}{2}} \Delta x}{\sqrt{2\pi \Delta_i}} , \frac{K_4 \Delta x}{ \pi \Delta_i} \right) \sup_k e_{i-1,k} \nonumber
\end{align}where $C_0>0$ and $K>0$ is the Lipschitz constant of the driver
$f$. Consequently, \begin{align} \sup_k e_{i,k} & \leq  C_0 \sup_{i,k}{ E}_{i,k} + \left({1+ 2\Delta_iK} \right) \max\left(\frac{K_4^{\frac{1}{2}} \Delta x}{\sqrt{2\pi \Delta_i}} , \frac{K_4 \Delta x}{ \pi \Delta_i} \right)  \sup_k e_{i-1,k} \nonumber\\
& \leq  C_0 \sup_{i,k}{ E}_{i,k} + ({1+ 2\Delta_i K}) \sup_k e_{i-1,k}
\label{eq:2convprf} \end{align} since \begin{equation*} \sup_i \max\left( \frac{K_4^{\frac{1}{2}} \Delta x}{\sqrt{2\pi \Delta_i}} , \frac{K_4 \Delta x}{ \pi \Delta_i} \right) \leq 1. \end{equation*}and
the Gronwall's Lemma yields \begin{equation} \sup_k e_{i,k} \leq C_0 e^{2T K} \sup_{i,k} { E}_{i,k} \label{eq:2prf4} \end{equation}
from the inequality of equation (\ref{eq:2convprf}) for $i=0,1,\ldots,n$
knowing that $e_{0,k}=0$. The last equation establishes the stability
of the Fourier interpolation method for the approximate solution $u_{i}$
since its error at any time step is absolutely bounded.

The inequalities of equations (\ref{eq:2prf12}), (\ref{eq:2prf3}) and (\ref{eq:2prf4})
lead to \begin{align} \sup_k \dot{e}_{i,k} & \leq \left( C_1 +  \frac{\Delta x}{\pi \Delta_i}  {C_0 e^{2T K}} \right) \sup_{i,k}E_{i,k} \nonumber\\
& \leq \left( C_1 +  {C_0 e^{2T K}} \right) \sup_{i,k}E_{i,k} \label{eq:2prf5}
\end{align}for a positive constant $C_{1}>0$. Hence, the convolution method
is also stable for the approximate gradient $\dot{u}_{i}$. The result of equation (\ref{eq:GDR-1}) follows by taking the supremum
on the left hand sides of equations (\ref{eq:2prf4}) and (\ref{eq:2prf5})
other time steps and applying Corollary \ref{cor:LDerr}.

\end{proof}
As for most explicit methods for PDEs, the convolution method requires
a stability condition as described in equation (\ref{eq:ccond-1}). In
general, Theorem \ref{thm:GDR} shows that the convolution method
converges when the space discretization is relatively as fine as the
time discretization. Other numerical methods for BSDEs, and particularly
Monte Carlo based methods, have a stability and convergence condition.
Indeed, error explosion occurs for fine time discretizations in the
backward methods of  \cite{gobet:2005} and \cite{bouchardtouzi:2004}. In order to maintain stability and
convergence, the space discretization has to be refined by increasing
the number of simulated paths.

\section{Higher order time discretization for FBSDEs \label{sec:Ext}}

In this section, we discuss further extensions of the Fourier interpolation
method on the alternative grid. In particular, we apply the Fourier interpolation method to
Runge-Kutta schemes for FBSDEs proposed by \cite{Chasscrisan:2013}.

\subsection{Runge-Kutta schemes}

The FBSDE of equation (\ref{eq:bsde}) is discretized on the time partition $\pi$.
Let $q\in\mathbb{N}^{*}$, we consider the $q$-stage Runge-Kutta
scheme giving the following numerical solution at mesh time $t_{i}$
\begin{align} {Z}^{\pi}_{t_i} &=  \CEsp{}{ H^{\varphi_1}_{t_i,\Delta_i} {Y}^{\pi}_{t_{i+1}}   + \Delta_i \sum_{j=1}^{q} \beta_j H^{{\bf \varphi_1}}_{t_i, (1 - \gamma_j)\Delta_i} f(t_{i,j}, X^{\pi}_{t_{i,j}} ,Y^{\pi}_{i,j}, Z^{\pi}_{i,j} )  }{t_i} \label{eq:zrk}\\
{Y}^{\pi}_{t_i} &= \CEsp{}{ {Y}^{\pi}_{t_{i+1}} +  \Delta_i \sum_{j=1}^{q+1} \alpha_j f(t_{i,j}, X^{\pi}_{t_{i,j}} ,Y^{\pi}_{i,j}, Z^{\pi}_{i,j} ) }{t_i} \label{eq:yrk} \end{align}for a set positive coefficients $\{\gamma_{j}\}_{j=1}^{q+1}$ such
that $0=\gamma_{1}<\ldots<\gamma_{q+1}=1$. The intermediate solutions
$\{(Y_{i,j}^{\pi},Z_{i,j}^{\pi})\}_{j=2}^{q}$ take the form \begin{align} {Z}^{\pi}_{i,j} &=  \CEsp{}{ H^{\varphi_j}_{t_{i,j}, \gamma_j\Delta_i} {Y}^{\pi}_{t_{i+1}}   + \Delta_i \sum_{k=1}^{j-1} \beta_{jk} H^{{\varphi_j}}_{t_{i,j}, (\gamma_j - \gamma_k) \Delta_i}f(t_{i,k}, X^{\pi}_{t_{i,k}} ,Y^{\pi}_{i,k}, Z^{\pi}_{i,k} )   }{{t_{i,j}} } \nonumber\\ \label{eq:zirk}\\
{Y}^{\pi}_{i,j} &= \CEsp{}{ {Y}^{\pi}_{t_{i+1}} +  \Delta_i \sum_{k=1}^{j} \alpha_{jk} f(t_{i,k}, X^{\pi}_{t_{i,k}} ,Y^{\pi}_{i,k}, Z^{\pi}_{i,k} ) }{{t_{i,j}}} \label{eq:yirk}\end{align}where \begin{equation} t_{i,j} = t_i + (1- \gamma_j ) \Delta_i \text{, } 1 \leq j \leq q+1 \end{equation} 
with $(Y_{i,1}^{\pi},Z_{i,1}^{\pi})=(Y_{t_{i+1}}^{\pi},Z_{t_{i+1}}^{\pi})$,
$(Y_{i,q+1}^{\pi},Z_{i,q+1}^{\pi})=(Y_{t_{i}}^{\pi},Z_{t_{i}}^{\pi})$
and terminal condition \begin{equation}(Y_{t_{n}},Z_{t_{n}})=(g(X_{T}),\sigma^{*}(T,X_{T})\nabla g(X_{T})). \end{equation}
The coefficients $\{\alpha_{j}\}_{j=1}^{q+1}$, $\{\beta_{j}\}_{j=1}^{q}$,
$\{\alpha_{jk}:1\leq j\leq q,1\leq k\leq j\}$ and $\{\beta_{jk}:1\leq j\leq q,1\leq k<j\}$
are all positive and satisfy \begin{align} \sum_{j=1}^{q+1} \alpha_j  &=  1  \\
\beta_{jj} &= 0, \quad 1 \leq j \leq q \text{,}\\
\sum_{k=1}^{j} \alpha_{jk}  &=  \sum_{k=1}^{j-1} \beta_{jk} = \gamma_j, \quad  1 < j \leq q. 
\end{align}Let $\mathcal{B}^{m}$ denote the set of continuous and bounded functions
on $[0,1]$ such that \begin{equation} \mathcal{B}^m := \left\{ \phi \in \mathcal{C}_{b} :  \int_{0}^1 s^k \phi(s)ds = \delta_{0,k}, k \leq m \text{ and } k,m \in \mathbb{N}^*  \right\}. \end{equation}The
stochastic coefficient $H_{t,\Delta}^{\varphi}$ with $t\in[0,T)$
and $\Delta>0$ is defined as \begin{equation} H^{\varphi}_{t,\Delta} :=\frac{1}{\Delta} \int_t^{t+\Delta} \varphi \left( \frac{s-t}{\Delta} \right) dW_s \end{equation}with
$\varphi\in\mathcal{B}^{m}$ for some $m\in\mathbb{N}^{*}$. 

The global error of the $q-$stage Runge-Kutta scheme $\mathcal{E}_{\pi}$
is defined as
\begin{align} \mathcal{E}_{\pi}^2 & :=  \max_{0 \leq i < n} \normf{ {Y_{t_i} - Y^{\pi}_{t_i}}}{L^2}^2 + \sum_{i=0}^{n-1} \Delta_i \normf{ {Z_{t_i} - Z^{\pi}_{t_i}}}{L^2}^2 \nonumber\\
&= \max_{0 \leq i < n} \CEsp{}{ \norm{Y_{t_i} - Y^{\pi}_{t_i}}^2}{} + \sum_{i=0}^{n-1} \Delta_i \CEsp{}{ \norm{Z_{t_i} - Z^{\pi}_{t_i}}^2}{} \label{eq:Erk} \end{align}and is hence weaker than the error $E_{\pi}$ considered for the Euler
scheme. Nonetheless, the global error $\mathcal{E}_{\pi}$ is easier
to handle since it is strongly related to the local time discretization
error which simplifies the theoretical study in \cite{Chasscrisan:2013}.

The scheme can be represented by the following tableau%

\begin{center} 
\begin{tabular}{c| c c c c c| c c c c}
$\gamma_1$    & $\alpha_{1,1}$ & $0$           & \ldots & $0$         & $0$ & $\beta_{1,1}$ & $0$           & \ldots & $0$       \\
$\gamma_2$        & $\alpha_{2,1}$ & $\alpha_{2,2}$ & \ldots & $0$        & $0$ & $\beta_{2,1}$ &  $\beta_{2,2}$ & \ldots & $0$    \\
$\vdots$   & $\vdots$ & $\vdots$ & $\ddots$ & $\vdots$        & $\vdots$ & $\vdots$       & $\vdots$      & $\ddots$ & $\vdots$     \\
$\gamma_q$ & $\alpha_{q,1}$ & $\alpha_{q,2}$ & \ldots & $\alpha_{q,q}$ & $0$ & $\beta_{q,1}$ & $\beta_{q,2}$ & \ldots & $\beta_{q,q}$\\ \hline
$\gamma_{q+1} $ & $\alpha_1$ & $\alpha_2$ & \ldots & $\alpha_q$ & $\alpha_{q+1}$ & $\beta_1$ & $\beta_2$ & \ldots & $\beta_{q}$  \\
\end{tabular} 

\end{center} 
One can observe that if $\alpha_{q+1}=0$ and $\alpha_{jj}=0$, $1<j\leq q$,
then the $q$-stage Runge-Kutta scheme is explicit. Otherwise, the
scheme is implicit. For instance, the Runge-Kutta schemes with tableau
\begin{center} 
\begin{tabular}{c| c c| c }
$0$ & $0$ & $0$ & $0$ \\ \hline
$1$ & $0$ & $1$ & $1$ \\
\end{tabular} 

\end{center} 
and the scheme with tableau %
\begin{center} 
\begin{tabular}{c| c c| c }
$0$ & $0$ & $0$ & $0$ \\ \hline
$1$ & $\frac{1}{2}$ & $\frac{1}{2}$ & $1$ \\
\end{tabular} 

\end{center} 
known as the Crank-Nicolson scheme constitute $1-$stage implicit
Runge-Kutta schemes. The only $1-$stage explicit Runge-Kutta scheme
admits the tableau %
\begin{center} 
\begin{tabular}{c| c c| c }
$0$ & $0$ & $0$ & $0$ \\ \hline
$1$ & $1$ & $0$ & $1$ \\
\end{tabular} 

\end{center} 

In \cite{Chasscrisan:2013} the implicit and
the explicit $1-$stage Runge-Kutta schemes are shown to be at least one-half
$(\frac{1}{2})$ order convergent. The Crank-Nicolson scheme, already
studied in  \cite{crisanmol:2013}, presents
a first-order of convergence. Notice that the Euler schemes used in
the previous chapters are not $1-$stage Runge-Kutta schemes since
they do not lead to any consistent tableau. Nonetheless, their structure
is equivalent to the explicit $1-$stage Runge-Kutta scheme and both
schemes display the same half $(\frac{1}{2})$ order of convergence.
The following tableau gives a example of explicit $2$-stage Runge-Kutta
schemes of first-order of convergence for $\gamma_{2}\in(0,1]$ and
$\beta_{1}\in[0,1]$. %

\begin{center} 
\begin{tabular}{c| c c c| c c}
$0$    & $0$ & $0$    & $0$         & $0$ & $0$  \\
$\gamma_2$ & $\gamma_2$ & $0$       & $0$ & $\gamma_2$ & $0$    \\ \hline
$1$ & $1-\frac{1}{2 \gamma_2}$ & $\frac{1}{2 \gamma_2}$ & $0$ & $\beta_1$ & $1 - \beta_1$ \\
\end{tabular} 

\end{center} 

\subsection{Further simplification}
From the $q$-stage Runge-Kutta scheme for BSDEs, one notices that
we have at least $2q$ conditional expectations to compute at each
time step. These conditional expectations can be simplified and made
more suitable for numerical implementation if we consider a reasonable
time discretization of the forward SDE. Hence, we make the following
assumption.%
\begin{assump}[Forward process discretization] \label{assump:tdisc} 
The following are assumed throughout this section.
\begin{enumerate} 
\item The forward SDE is discretized with the piecewise constant process $X^{\pi}$ such that for $t\in [t_i,t_{i+1})$ we have $ X^{\pi}_t = X^{\pi}_{t_i} $ pathwise. 
\item The forward SDE time discretization with global error $\mathcal{E}_{X,\pi}$ is of order $m>0$ i.e \begin{equation} \mathcal{E}_{X,\pi}^2 := \max_{0\leq i \leq n} \normf{ {X_{t_i} - X^{\pi}_{t_i} }  }{L^2}^2 = \mathcal{O}( \norm{\pi}^{2 m}). \end{equation} 

\item The forward SDE time discretization admits the conditional characteristic functions $\phi_i : \mathbb{R}^d \times \mathbb{R}^d \rightarrow \mathbb{C}$   \begin{equation}  \phi_i(\nu , x) = \CEsp{}{e^{ {\bf i} \nu^* \left( X^{\pi}_{t_{i+1}} - X^{\pi}_{t_i} \right)  } | X^{\pi}_{t_i} = x }{}\end{equation} and $\Phi_{i,j} : \mathbb{R}^d \times \mathbb{R}^d \rightarrow \mathbb{C}^d$ \begin{equation}  \Phi_{i,j}(\nu , x) = \CEsp{}{H^{\varphi_j}_{t_{i,j}, \gamma_j \Delta_i} e^{ {\bf i} \nu^* \left( X^{\pi}_{t_{i+1}} - X^{\pi}_{t_i} \right)  } | X^{\pi}_{t_i} = x }{} \label{eq:Phidef2}\end{equation} for $0 \leq i < n$ and $1 < j \leq q+1$ with $\varphi_{q+1} = \varphi_1$.

\item There are positive constants $p_0$, $q_0$, $s_0$ ,$K_0$ and $C_0>0$ such that 
  \begin{equation}
    \max \left( \inf_{ s \in \mathbb{R}_d^+ } e^{-s* t}\phi_i({\bf i}s,x) , \inf_{s \in \mathbb{R}_d^+ } e^{-s* t}\phi_i(-{\bf i}s,x) \right)  \leq e^{- K_0 \Delta_i^{-s_0} \norm{t}^{q_0}}, 
\end{equation} 
$\forall t \in \mathbb{R}_d^+$, hence the discrete version of the forward process has conditional exponential moments. In addition, 
\begin{equation} \int_{\mathbb{R}^d} \norm{\phi_i(\nu,x)} d\nu + \max_{1 < j \leq q+1}\int_{\mathbb{R}^d} \norm{\Phi_{i,j}(\nu,x)} d \nu  \leq C_0 \Delta_i ^ {-p_0} .  \end{equation}
\end{enumerate}
\end{assump}

It{\^o}-Taylor expansion
based schemes are an example of SDE discretization satisfying the conditions of Assumption \ref{assump:tdisc}. A more complete presentation of these
schemes can be found in \cite{kloden:1992}.  The next theorem gives a simplification of the BSDE time discretization
expressions.

\begin{thm}
\label{thm:simprkexp}Under Assumption \ref{assump:tdisc} (1), the
solution of the $q$-stage Runge-Kutta scheme satisfies \begin{equation} \{(Y_{i,j}^{\pi},Z_{i,j}^{\pi})\}_{j=2}^{q+1} \in \mathcal{F}_{t_i} \end{equation}
for $0\leq i<n$. Consequently, we can write \begin{align} Z^{\pi}_{i,j} &= \CEsp{}{ H^{\varphi_j}_{t_{i,j}, \gamma_j \Delta_i} \left( Y^{\pi}_{t_{i+1}} + \Delta_i \beta_{j,1}f(t_{i+1}, X^{\pi}_{t_{i+1}}, Y^{\pi}_{t_{i+1}}, Z^{\pi}_{t_{i+1}} )   \right) }{t_i}  \label{eq:zrk2} \\
Y^{\pi}_{i,j} &= \CEsp{}{  Y^{\pi}_{t_{i+1}} + \Delta_i \alpha_{j,1} f(t_{i+1}, X^{\pi}_{t_{i+1}}, Y^{\pi}_{t_{i+1}}, Z^{\pi}_{t_{i+1}} )  }{t_i} \nonumber \\
& + \Delta_i \sum_{k=2}^{j} \alpha_{jk} f(t_{i,k}, X^{\pi}_{t_{i}}, Y^{\pi}_{i,k}, Z^{\pi}_{i,k} )   \label{eq:yrk2}
\end{align} for $0\leq i<n$ and $1<j\leq q+1$ where $\varphi_{q+1}=\varphi_{1}$,
$\beta_{q+1,1}=\beta_{1}$ and $\alpha_{q+1,k}=\alpha_{k}$.\end{thm}
\begin{proof}
Clearly $(Y_{i,q+1}^{\pi},Z_{i,q+1}^{\pi})=(Y_{t_{i}}^{\pi},Z_{t_{i}}^{\pi})\in\mathcal{F}_{t_{i}}$
from equations (\ref{eq:zrk}) and (\ref{eq:yrk}). For $1<j\leq q$
and $0\leq i<n$, we have \begin{align*} 
{Y}^{\pi}_{i,j} &= \CEsp{}{ {Y}^{\pi}_{t_{i+1}} +  \Delta_i \sum_{k=1}^{j} \alpha_{jk} f(t_{i,k}, X^{\pi}_{t_{i,k}},Y^{\pi}_{i,k}, Z^{\pi}_{i,k} ) 
\left. \right| X^{\pi}_{t_{i,j} }  }{}  \tag{ starting from equation \ref{eq:yirk}} \\
&= \CEsp{}{ {Y}^{\pi}_{t_{i+1}} +  \Delta_i \sum_{k=1}^{j} \alpha_{jk} f(t_{i,k}, X^{\pi}_{t_{i,k}},Y^{\pi}_{i,k}, Z^{\pi}_{i,k} ) \left. \right| X^{\pi}_{t_i}  }{}
\tag{by Assumption \ref{assump:tdisc} since $t_{i,j} \in [t_i, t_{i+1})$} \\
&= \CEsp{}{ {Y}^{\pi}_{t_{i+1}} +  \Delta_i \sum_{k=1}^{j} \alpha_{jk} f(t_{i,k}, X^{\pi}_{t_{i,k}},Y^{\pi}_{i,k}, Z^{\pi}_{i,k} )}{t_i}
\end{align*}so that $Y_{i,j}^{\pi}\in\mathcal{F}_{t_{i}}$. Similar arguments
also show that $Z_{i,j}^{\pi}\in\mathcal{F}_{t_{i}}$ starting from
equation (\ref{eq:zirk}).

Since$\{(Y_{i,j}^{\pi},Z_{i,j}^{\pi})\}_{j=2}^{q+1}\in\mathcal{F}_{t_{i}}$,
we naturally get equation (\ref{eq:yrk2}) from equations (\ref{eq:yirk}) and (\ref{eq:yrk}) .
In addition, knowing that \begin{equation} \CEsp{}{ H^{\varphi_j }_{t_{i,j}, (\gamma_i - \gamma_k)\Delta_i} }{t_i} = 0 \text{ , } 1 < k < j \end{equation}
leads to equation (\ref{eq:zrk2}) from equations (\ref{eq:zirk}) and (\ref{eq:zrk}).
\end{proof}
As a consequence of Assumption \ref{assump:tdisc}, if the $q-$stage
Runge-Kutta scheme and the forward SDE time discretization are of
order $m>0$ then error of the FBSDE numerical solution defined as
$\mathcal{E}_{X,\pi}+\mathcal{E}_{\pi}$ is of order $m$. We must
hence choose the Runge-Kutta scheme and the SDE scheme accordingly.

\subsection{Fourier representation}

Following Theorem \ref{thm:simprkexp}, the intermediate solutions
$\{(u_{i,j},\dot{u}_{i,j})\}_{j=2}^{q+1}$ at mesh time $t_{i}$,
$0\leq i<n$, are given by \begin{align} 
\dot{u}_{i,j}(x) &= \CEsp{}{ H^{\varphi_j}_{t_{i,j}, \gamma_j \Delta_i} \tilde{u}_{i+1}( W_{t_{i+1}} , \beta_{j,1}) | X^{\pi}_{t_i} = x     }{} \label{eq:zrk3}\\
{u}_{i,j}(x) &= \CEsp{}{ \tilde{u}_{i+1}( W_{t_{i+1}} , \alpha_{j,1}) | X^{\pi}_{t_i} = x     }{} + \Delta_i  \sum_{k=2}^{j} \alpha_{jk} f(t_{i,k}, x, u_{i,k}(x), \dot{u}_{i,k}(x) ) \label{eq:yrk3}
\end{align}for $1<j\leq q+1$ with $\varphi_{q+1}=\varphi_{1}$, $\beta_{q+1,1}=\beta_{1}$
and $\alpha_{q+1,k}=\alpha_{k}$ . The approximate solution $u_{i}$
and approximate gradient $\dot{u}_{i}$ at mesh time $t_{i}$, $0\leq i<n$,
are then \begin{align} u_i(x) = u_{i,q+1}(x)  \label{eq:rksol1} \\
\dot{u}_i(x) = \dot{u}_{i,q+1}(x) \label{eq:rksol2}
\end{align}with \begin{equation} \tilde{u}_{i+1}(x , \alpha )= u_{i+1}(x) + \Delta_i \alpha f(t_{i+1}, x, u_{i+1}(x) , \dot{u}_{i+1}(x) )\end{equation}
and \begin{align} u_n(x) &= g(x)  \\
\dot{u}_n(x) &= \sigma^*(T, x) \nabla g(x).
\end{align}

In this setting, we have that
\begin{equation}
   {u}_{i,j}(x) = \CEsp{}{ \tilde{u}_{i+1}( X^{\pi}_{t_{i+1}} , \alpha_{j,1}) | X^{\pi}_{t_i} = x     }{} + \Delta_i  \sum_{k=2}^{j} \alpha_{jk} f(t_{i,k}, x, u_{i,k}(x), \dot{u}_{i,k}(x) ) \label{eq:yrk5-1}
\end{equation}
Note that
\begin{align}
&  \CEsp{}{ \tilde{u}_{i+1}( X^{\pi}_{t_{i+1}} , \alpha_{j,1}) | X^{\pi}_{t_i} = x     }{} =
  \CEsp{x}{ \frac{1}{(2\pi)^d } \int_{\mathbb{R}^d} e^{ {\bf i} \nu^* X^{\pi}_{t_{i+1}} } \mathfrak{F}\left[\tilde{u}_{i+1}(., \alpha_{j,1})\right](\nu)   d\nu }{t_i} \nonumber\\
  &= \frac{1}{(2\pi)^d} \int_{\mathbb{R}^d} \CEsp{x}{  e^{ {\bf i} \nu^* X^{\pi}_{t_{i+1}} }     }{t_i} \mathfrak{F}\left[\tilde{u}_{i+1}(., \alpha_{j,1})\right](\nu)  d\nu   \tag{using Fubini's theorem} \\
  &=  \frac{1}{(2\pi)^d} \int_{\mathbb{R}^d} e^{ {\bf i} \nu^* x } \phi_i(\nu,x) \mathfrak{F}\left[\tilde{u}_{i+1}(., \alpha_{j,1})\right](\nu)  d\nu. \label{eq:yrk5-2}
  \end{align}
Therefore, by (\ref{eq:yrk5-1}) and (\ref{eq:yrk5-1}), we have
  \begin{align}
   {u}_{i,j}(x) &=
   \mathfrak{F}^{-1} \left[ \mathfrak{F}\left[\tilde{u}_{i+1}(., \alpha_{j,1})\right](\nu)  \phi_i( \nu, x )\right](x) \nonumber \\
   &+ \Delta_i \sum_{k=2}^{j } \alpha_{jk} f(t_{i,k}, x,u_{i,k}(x), \dot{u}_{i,k}(x) ) \label{eq:yrk5}
  \end{align}
whenever $\tilde{u}_{i+1}(.,\alpha)$ is Lebesgue integrable. 
  
As to the intermediate solutions $\dot{u}_{i,j}$, $0\leq i<n$ and
$1<j\leq q+1$, we have \begin{align} 
 {\dot{u}_{i,j}(x)} &= \CEsp{}{ H^{\varphi_j}_{t_{i,j}, \gamma_j \Delta_i} \tilde{u}_{i+1}( X^{\pi}_{t_{i+1}} , \beta_{j,1}) | X^{\pi}_{t_i} = x     }{} \nonumber \\
&= \CEsp{x}{H^{\varphi_j}_{t_{i,j}, \gamma_j \Delta_i}  \frac{1}{(2\pi)^d} \int_{\mathbb{R}^d} e^{ {\bf i} \nu^* X^{\pi}_{t_{i+1}} } \mathfrak{F}\left[\tilde{u}_{i+1}(., \beta_{j,1})\right](\nu)   d\nu }{t_i} \nonumber \\
&=\frac{1}{(2\pi)^d} \int_{\mathbb{R}^d} \CEsp{x}{H^{\varphi_j}_{t_{i,j}, \gamma_j \Delta_i}   e^{ {\bf i} \nu^* X^{\pi}_{t_{i+1}} }  }{t_i} \mathfrak{F}\left[\tilde{u}_{i+1}(., \beta_{j,1})\right](\nu)   d\nu  \tag{using Fubini's theorem} \nonumber\\
&=\frac{1}{(2\pi)^d} \int_{\mathbb{R}^d} e^{ {\bf i} \nu^* x } \Phi_{i,j}(\nu,x)  \mathfrak{F}\left[\tilde{u}_{i+1}(., \beta_{j,1})\right](\nu)   d\nu \nonumber \\
&= \mathfrak{F}^{-1} \left[ \mathfrak{F}\left[\tilde{u}_{i+1}(., \alpha_{j,1})\right](\nu)  \Phi_{i,j}( \nu, x )\right](x)  \label{eq:zrk5}
\end{align}for an integrable function $\tilde{u}_{i+1}(.,\alpha)$.

Even if the expressions in equations (\ref{eq:yrk5}) and (\ref{eq:zrk5})
appear too general, they are implementable with the Fourier interpolation
method for $d=1$ in various particular cases. One can retrieve the characteristics
$\phi_{i}$ and $\Phi_{i,j}$ and also perform the corrections due to the transform of equation (\ref{eq:trans2}) for many SDE time discretizations.  The following lemma helps in retrieving the conditional characteristics.

\begin{lem}
The conditional characteristics $\Phi_{i,j}$ write
\begin{equation} 
{\Phi_{i,j}(\nu, x)}  =    \CEsp{x}{  {\bf H}^*_{{i,j}}  {\bf i}\nu   e^{ {\bf i} \nu^* \left( X^{\pi}_{t_{i+1}} - X^{\pi}_{t_i} \right)}  }{t_i}  \label{eq:zrk51}
\end{equation}
with \begin{equation} {\bf H}_{i,j} = \frac{1}{\gamma_j \Delta_i} \int_{t_{i,j}}^{t_{i+1}} D_s X^{\pi}_{t_{i+1}} \varphi_j \left( \frac{s-t_{i,j}}{\gamma_j \Delta_i} \right)ds  
\label{eq:zrk51p} \end{equation} where $D_s X^{\pi}_{t_{i+1}}$ is the Malliavin derivative of $X_{t_{i+1}}^{\pi}$
given $X_{t_{i}}^{\pi}=x$
\end{lem}

\begin{proof}
The lemma is proved by applying the duality formula and the chain rule successively to equation (\ref{eq:Phidef2}).
\end{proof}

\subsubsection{Half-order It{\^o}-Taylor schemes}
The Euler scheme constitutes the main example of half-order It{\^o}-Taylor
scheme with step 
\begin{equation*} X^{\pi}_{t_{i+1}} =  X^{\pi}_{t_{i}} + a(t_i, X^{\pi}_{t_{i}})\Delta_i + \sigma(t_i, X^{\pi}_{t_{i}}) \Delta W_i.
\end{equation*}
In addition, we
have that  $D_s \Delta_i = {\bf 0}_{d \times 1}$ and $D_s \Delta W_i = { \bf I}_{d \times d}$ for $s\in(t_{i},t_{i+1})$ where ${\bf 0}$ and ${\bf I}$ are the zero matrix and the identity matrix respectively.
Hence, \begin{equation*} D_s X^{\pi}_{t_{i+1}} = \sigma(t_i, X^{\pi}_{t_{i}}) \end{equation*}
so we get, from equation (\ref{eq:zrk51}), that \begin{align} {\Phi_{i,j}(\nu, x)} &=  \sigma^*(t_i, x) {\bf i } \nu \CEsp{x}{ e^{ {\bf i} \nu^* \left( X^{\pi}_{t_{i+1}} - X^{\pi}_{t_i} \right)  } }{t_i}  \text{ (since $\varphi_j \in \mathcal{B}^0$),} \nonumber \\
&=    \sigma^*(t_i, x) {\bf i } \nu \phi_i(\nu,x). \label{eq:zrk5half}
\end{align}The conditional characteristic function is explicitly given by \begin{equation} \phi_i( \nu, x ) = \exp\left \{ \Delta_i \left(  {\bf i}\nu^*a(t_i, x) - \frac{1}{2} \nu^* \sigma^2(t_i, x) \nu     \right) \right \} \label{eq:yrk5half} \end{equation}since
the increment has a Gaussian distribution. 

Equations (\ref{eq:yrk5}) and (\ref{eq:zrk5}) along with the characteristics
of equations (\ref{eq:yrk5half}) and (\ref{eq:zrk5half}) define
the Fourier method under half-order It{\^o}-Taylor schemes and
the method is implementable in one dimension ($d=1$) with the procedure given in section \ref{sec:implement}.
The following theorem generalizes the result of Theorem \ref{lem:altT}
to Runge-Kutta schemes under half-order It{\^o}-Taylor schemes.
\begin{thm}
Let $\tilde{u}^{\alpha, \beta}_{i+1}(.,y)$ be the alternative transform
defined in equation (\ref{eq:trans2}) of the approximate solution
$\tilde{u}_{i+1}(.,y)$. Then the intermediate solutions $u_{i,j}$
and $\dot{u}_{i,j}$ in equations (\ref{eq:yrk5}) and (\ref{eq:zrk5})
satisfy \begin{align} {u}_{i,j}(x) &= \mathfrak{F}^{-1}[ \mathfrak{F}[\tilde{u}^{\alpha,\beta}_{i+1}(.,\alpha_{j,1})](\nu) \phi_i(\nu,x) ](x) \nonumber \\
& -  \alpha [ (x + \Delta_i a(t_i, x))^2 + \Delta_i \sigma^2(t_i, x)] - \beta (x + \Delta_i a(t_i, x))  \nonumber\\
& + \Delta_i  \sum_{k=2}^{j} \alpha_{jk} f(t_{i,k}, x, u_{i,k}(x), \dot{u}_{i,k}(x) ) \\
\dot{ u}_{i,j} (x) &= \sigma(t_i, x) \mathfrak{F}^{-1}[ \mathfrak{F}[\tilde{u}^{\alpha,\beta}_{i+1}(.,\beta_{j,1})](\nu) {\bf i} \nu \phi_i(\nu,x) ](x) \nonumber \\
& - \sigma(t_i, x) [ 2 \alpha(x + \Delta_i a(t_i, x))  + \beta]. \end{align}under a half-order It{\^o}-Taylor scheme when $d=1$.
\end{thm}

\subsubsection{First-order It{\^o}-Taylor schemes}

Consider the first-order scheme
\begin{equation*}
  X^{\pi}_{t_{i+1}} =  X^{\pi}_{t_{i}} + a(t_i, X^{\pi}_{t_{i}})\Delta_i + \sigma(t_i, X^{\pi}_{t_{i}}) \Delta W_i
  + \sigma^2(t_i, X^{\pi}_{t_{i}}) \int_{t_i}^{t_{i+1}} \int_{t_i}^t dW_u dW_t.
\end{equation*}
Then knowing that
\begin{equation}
  D_s \int_{t_i}^{t_{i+1}} \int_{t_i}^t dW_u dW_t = {\bf D}( \Delta W_i ),\label{eq:mal2} 
\end{equation}
using the fundamental theorem of calculus where ${\bf D}( x )$ is the diagonal matrix composed with the elements of $x$, for $s\in(t_{i},t_{i+1})$, the Malliavin derivative of the discretized forward process is given
by
\begin{equation}
  D_s X^{\pi}_{t_{i+1}} = \sigma(t_i, x) + \sigma^2(t_i, x){\bf D}( \Delta W_i ).
\end{equation} 
Equation (\ref{eq:zrk51}) leads to
\begin{align}
  {\Phi_{i,j}(\nu, x)} 
&=  \sigma^*(t_i, x) {\bf i } \nu \CEsp{x}{ e^{ {\bf i} \nu^* \left( X^{\pi}_{t_{i+1}} - X^{\pi}_{t_i} \right)  }  }{t_i} \nonumber \\ &+   \CEsp{x}{ {\bf D}( \Delta W_i ) {\sigma^2(t_i, x)}{\bf i } \nu  e^{ {\bf i} \nu^* \left( X^{\pi}_{t_{i+1}} - X^{\pi}_{t_i} \right)  }   }{t_i} \tag{since $\varphi_j \in \mathcal{B}^0$} \nonumber \\
&=   \sigma^*(t_i, x) {\bf i } \nu \phi_i(\nu,x)  +     \CEsp{x}{ {\bf D}( \sigma^2(t_i, x) {\bf i }\nu  ) \Delta W_i   e^{ {\bf i} \nu^* \left( X^{\pi}_{t_{i+1}} - X^{\pi}_{t_i} \right)  } }{t_i} \nonumber  \\
&= ({\bf I} -  {\bf D}( \sigma^2(t_i, x) {\bf i }\nu  ) \Delta_i)^{-1} {\sigma^*(t_1,x)}  {\bf i}\nu {\phi}_{i} (\nu,x) \label{eq:zrkint}
\end{align}
since
\begin{align*}
  \CEsp{x}{ \Delta W_i  e^{ {\bf i} \nu^* \left( X^{\pi}_{t_{i+1}} - X^{\pi}_{t_i} \right)  }  }{t_i} &=  \sigma^*(t_i, x) \Delta_i \CEsp{x}{  {\bf i} \nu e^{ {\bf i} \nu^* \left( X^{\pi}_{t_{i+1}} - X^{\pi}_{t_i} \right)  } }{t_i}  \\ &+ {\bf D}(  \sigma^2(t_i, x) {\bf i }\nu  ) \Delta_i \CEsp{x}{   \Delta W_i e^{ {\bf i} \nu^* \left( X^{\pi}_{t_{i+1}} - X^{\pi}_{t_i} \right)  } }{t_i}
\end{align*}
using the duality formula, so that 
\begin{align*} \CEsp{x}{ \Delta W_i  e^{ {\bf i} \nu^* \left( X^{\pi}_{t_{i+1}} - X^{\pi}_{t_i} \right)  }  }{t_i}  &=   { \zeta}_i(\nu, x) \Delta_i\sigma^*(t_i, x) \CEsp{x}{ {\bf i} \nu   e^{ {\bf i} \nu^* \left( X^{\pi}_{t_{i+1}} - X^{\pi}_{t_i} \right)  }  }{t_i} \nonumber \\
 &= \zeta_i(\nu, x) \Delta_i\sigma^*(t_i, x) {\bf i} \nu \phi_i(\nu,x) \label{eq:fp1}
  \end{align*} with 
\begin{equation}  
\zeta_{i}(\nu,x) = ( {\bf I} -  {\bf D}( \sigma^2(t_i, x) {\bf i }\nu  ) \Delta_i)^{-1}.
\end{equation}

As to the conditional characteristic $\phi_{i}$, it can be easily derived as
\begin{equation}
  \phi_i(\nu, x) = det( \zeta_{i}(\nu,x) )^{\frac{1}{2}} \exp \left( \frac{1}{2} {\bf i }\nu^*  \zeta^{-1}_{i}(\nu,x) {\bf 1}_{d \times 1}
  + {\bf i }\nu^* \kappa_{i}(x)\right)  \label{eq:char1or}\end{equation}
where
\begin{equation*} 
  \kappa_{i}(x) = a(t_i,x)\Delta_i - \frac{1}{2}(\sigma^2(t_i,x)\Delta_i + 1 ){\bf 1}_{d \times 1}
\end{equation*}
knowing that $X^{\pi}_{t_{i+1}} - X^{\pi}_{t_i}$, given $X^{\pi}_{t_i}=x$, is an affine function of a multivariate non-central $\chi^2$ random variable with $1$ degree of freedom and non-centrality parameters $1$.

Equations (\ref{eq:yrk5}) and (\ref{eq:zrk5})  along with the expressions
in equations (\ref{eq:zrkint}) and (\ref{eq:char1or}) characterize the method under first
order discretizations on the forward process. The procedure introduced in Section \ref{sec:implement}
allows us to do the necessary computations given the characteristics $\phi_{i}$
and $\Phi_{i,j}$ and using the following theorem.
\begin{thm}
Let $\tilde{u}^{\alpha, \beta}_{i+1}(.,y)$ be the alternative transform
defined in equation (\ref{eq:trans2}) of the approximate solution
$\tilde{u}_{i+1}(.,y)$. Then the intermediate solutions $u_{i,j}$
and $\dot{u}_{i,j}$ in equations (\ref{eq:yrk5}) and (\ref{eq:zrk5}) 
satisfy
\begin{align}
  {u}_{i,j}(x) &= \mathfrak{F}^{-1}[ \mathfrak{F}[\tilde{u}^{\alpha,\beta}_{i+1}(.,\alpha_{j,1})](\nu) \phi_i(\nu,x) ](x)  \nonumber \\
  &-  \alpha \left[ (x + \Delta_i a(t_i,x) )^2 + \Delta_i \sigma^2(t_i,x) + \frac{1}{2}\Delta^2_i \sigma^4(t_i, x) \right] \nonumber \\
  &- \beta (x + \Delta_i a(t_i, x) ) + \Delta_i  \sum_{k=2}^{j} \alpha_{jk} f(u_{i,k}(x), \dot{u}_{i,k}(x) ) \label{eq:rklemy} \\
  \dot{ u}_{i,j} (x) &= \sigma(t_i, x) \mathfrak{F}^{-1}[ \mathfrak{F}[\tilde{u}^{\alpha,\beta}_{i+1}(.,\beta_{j,1})](\nu) {\bf i} \nu  \zeta_i(\nu,x) \phi_i(\nu,x) ](x) \nonumber \\
  &- \sigma(t_i, x) \left[ 2 \alpha \left(x + \Delta_i a(t_i, x) + \Delta_i \sigma^2(t_i, x)  \right)  + \beta \right] \label{eq:rklemz}
\end{align}
under a first-order It{\^o}-Taylor scheme when $d=1$.
\end{thm}
\begin{proof}
By the definition of the alternative transform, we must have that
\begin{align}
  {u}_{i,j}(x) &= \mathfrak{F}^{-1}[ \mathfrak{F}[\tilde{u}^{\alpha,\beta}_{i+1}(.,\alpha_{j,1})](\nu) \phi_i(\nu,x) ](x)
  -  \CEsp{x}{ \alpha (X^{\pi}_{t_{i+1}})^2 + \beta X^{\pi}_{t_{i+1}} }{t_i} \nonumber \\
  & + \text{   } \Delta_i \mathbf{1}_{ \{j>2\} } \sum_{k=2}^{j-1} \alpha_{jk} f(u_{i,k}(x), \dot{u}_{i,k}(x) ). \label{eq:prk1}
\end{align}
Notice that
\begin{equation}
  \CEsp{x}{ X^{\pi}_{t_{i+1}} }{t_i} = x + \Delta_i a(t_i, x) \label{eq:prk2}
\end{equation}
and
\begin{align}
&  \CEsp{x}{ (X^{\pi}_{t_{i+1}})^2 }{t_i} = \CEsp{x}{ X^{\pi}_{t_{i+1}} }{t_i}^2 + {\bf Var}_{t_i}^{x}[X^{\pi}_{t_{i+1}}]  \nonumber\\
&= (x + \Delta_i a(t_i, x))^2  + \CEsp{x}{ \left( \sigma(t_i, x) \Delta W_i + \sigma^2(t_i, x) \int_{t_i}^{t_{i+1}}\int_{t_i}^t dW_u dW_t \right)^2 }{t_i} \nonumber\\
&= (x + \Delta_i \sigma(t_i, x) )^2   + \Delta_i \sigma^2(t_i, x) + \frac{1}{2}\Delta^2_i \sigma^4(t_i, x). \label{eq:prk3} 
\end{align}
Equations (\ref{eq:prk1}), (\ref{eq:prk2}) and (\ref{eq:prk3}) lead
to the expression for $u_{i,j}$ in equation (\ref{eq:rklemy}).

The definition of the alternative transform also requires
\begin{align}
  \dot{ u}_{i,j} (x) &= \sigma(t_i, x) \mathfrak{F}^{-1}[ \mathfrak{F}[\tilde{u}^{\alpha,\beta}_{i+1}(.,\beta_{j,1})](\nu)  {\bf i} \nu  \zeta_i(\nu,x) \phi_i(\nu,x) ](x)  \nonumber \\
  &- \CEsp{x}{ H^{\varphi_j}_{t_{i,j}, \gamma_j \Delta_i} ( \alpha (X^{\pi}_{t_{i+1}})^2 + \beta X^{\pi}_{t_{i+1}} ) }{t_i} \nonumber \\
  &= \sigma(t_i,x) \mathfrak{F}^{-1}[ \mathfrak{F}[\tilde{u}^{\alpha,\beta}_{i+1}(.,\beta_{j,1})](\nu)  {\bf i} \nu  \zeta_i(\nu,x) \phi_i(\nu,x) ](x)  \nonumber \\
  &- \sigma(t_i, x)  \CEsp{x}{ 2\alpha (X^{\pi}_{t_{i+1}}) + \beta }{t_i}  \nonumber \\
  &- \sigma^2(t_i, x)  \CEsp{x}{ \Delta W_i \left(2\alpha (X^{\pi}_{t_{i+1}}) + \beta \right) }{t_i}  \tag{using the duality formula} \nonumber\\
  &= \sigma(t_i, x) \mathfrak{F}^{-1}[ \mathfrak{F}[\tilde{u}^{\alpha,\beta}_{i+1}(.,\beta_{j,1})](\nu)  {\bf i} \nu  \zeta_i(\nu,x) \phi_i(\nu,x) ](x) \nonumber \\
  &- \sigma(t_i, x) \left[ 2 \alpha \left(x + \Delta_i a(t_i, x) + \Delta_i \sigma^2(t_i, x) \right)  + \beta \right] 
\end{align}
using the duality formula once again. 
\end{proof}

The implementation of higher order time discretization for FBSDEs on the alternative grid is described in the following algorithm. Algorithm \ref{alg:convHD} produces the numerical intermediate solutions $\{ {u}_{i,j,k} \}_{k=0}^{N_iN}$, $\{ {\tilde{u}}_{i,j,k} \}_{k=0}^{N_iN}$ and $\{ {\dot{u}}_{i,j,k} \}_{k=0}^{N_iN}$ at time step $t_i$, $ 0 \leq i < n$ and stage $j$, $1 \leq j \leq q+1$
for the approximate solution $u_{i}$, the intermediate solution $\tilde{u}_{i}$ and
the approximate gradient $\dot{u}_{i}$ respectively, $i=0,1,\ldots,n-1$.

\begin{algo}  \label{alg:convHD} Fourier Interpolation Method on Alternative Grid for $q$-stage Runge-Kutta schemes
\begin{enumerate}
\item Discretize the restricted real space $[x_0-\frac{N_nl}{2},x_0 + \frac{N_nl}{2}]$ and the restricted Fourier space $[-\frac{L}{2},\frac{L}{2}]$ with $N_nN$ space steps so to have the real space nodes $\{ x_{nk}\}_{k=0}^{N_nN}$ and $\{ \nu_{nk}\}_{k=0}^{N_nN}$
\item Value $u_n(x_{nk}) = g(x_{nk})$
\item For any $i$ from $n-1$ to $0$
   
   \begin{enumerate}   
      
     \item For any j,  $1 < j \leq q+1$

\begin{enumerate}
\item Compute $\alpha$ and $\beta$ defining the transform of equation (\ref{eq:trans2}), such that
\begin{equation} \theta_{i+1}  = \tilde{ u}_{i+1} ^{\alpha,\beta}(., \alpha_{j,1}) \end{equation} and $\theta_{i+1}$ satisfies the boundary conditions of equations (\ref{eq:bound1}) and (\ref{eq:bound2}).
\item Compute $\theta_i(x_{ik})$ through equation (\ref{eq:numconv2p}) for $k=0,1,\ldots,N_iN$ with \begin{equation} \psi(\nu, x) = \phi_i(\nu, x) \end{equation} and retrieve the values $\tilde{u}_{i,j,k}$ with the appropriate correction.

\item Compute $\alpha$ and $\beta$ defining the transform of equation (\ref{eq:trans2}), such that
\begin{equation} \theta_{i+1}  = \tilde{ u}_{i+1} ^{\alpha,\beta}(., \beta_{j,1}) \end{equation} and $\theta_{i+1}$ satisfies the boundary conditions of equations (\ref{eq:bound1}) and (\ref{eq:bound2}).

\item Compute $\theta_i(x_{ik})$ through equation (\ref{eq:numconv2p}) for $k=0,1,\ldots,N_iN$ with \begin{equation} \psi(\nu, x) = \Phi_{i,j}(\nu, x) \end{equation} and retrieve the values $\dot{u}_{i,j,k}$ with the appropriate correction.
\item Compute the values $u_{i,j,k}$ as 
\begin{equation} u_{i,j,k} = \tilde{u}_{i,j,k} + \Delta_i  \sum_{s=2}^{j} \alpha_{js} f(t_{i,s}, x_{i,k}, u_{i,s,k}, \dot{u}_{i,s,k} ) \end{equation} for $k=0,1,\ldots,N_iN$ through equation (\ref{eq:uf}).
\item Update the real space grid with equation (\ref{eq:intergrid}) and the Fourier space grid by discretizing the interval $[-\frac{L}{2},\frac{L}{2}]$ with $N_iN$ space steps so to have the real space nodes $\{ x_{ik}\}_{k=0}^{N_iN}$ and $\{ \nu_{ik}\}_{k=0}^{N_iN}$.
\end{enumerate}

    \item Set $u_{i,k} = u_{i,q+1,k}$ and $ \dot{u}_{i,k} = \dot{u}_{i,q+1,k}$

\end{enumerate}
\end{enumerate}
\end{algo}

\subsection{Spatial discretization error analysis \label{sec:ErrorB}}

We denote by $\{ {u}_{i,j,k} \}_{k=0}^{N_iN}$ and $\{ {\dot{u}}_{i,j,k} \}_{k=0}^{N_iN}$
the intermediate numerical solutions obtained at time step $t_{i}$,
$i=0,1,\ldots,n-1$ and stage $j$, $1<j\leq q+1$, from the Fourier interpolation method on the alternative grid when using a $q-$stage Runge-Kutta
scheme. In addition, $\{ \mathbf{u}_{i,j,k} \}_{k=0}^{N_iN}$ and
$\{ \mathbf{\dot{u}}_{i,j,k} \}_{k=0}^{N_iN}$ are the intermediate
numerical solutions obtained at the intermediate stage $j$, $1<j\leq q+1$,
of time step $t_i$ given the exact solutions $u_{i+1}$ and ${\dot u}_{i+1}$
at $t_{i+1}$. We have from the notation previously used that
the numerical solutions at $t_{i}$ write \begin{align} {u}_{i,k} &= {u}_{i,q+1,k} \\ \dot{u}_{i,k} &= \dot{u}_{i,q+1,k} \end{align}
and are computed from the intermediate solutions $\{ \tilde{u}_{i,k} \}_{k=0}^{N_iN}$,
$0<i\leq n$ where $\tilde{u}_{n,k} = \tilde{u}_{n}(x_{n,k})$. When
the exact solutions $u_{i+1}$ and ${\dot u}_{i+1}$ are known at
$t_{i+1}$, we also write \begin{align} \mathbf{u}_{i,k} &= \mathbf{u}_{i,q+1,k} \\ \mathbf{ \dot{u}}_{i,k} &= \mathbf{ \dot{u}}_{i,q+1,k}. \end{align}
The local (space) discretization error has the form \begin{equation} { E}_{ik} := \norm{u_i(x_k) - \mathbf{u}_{i,k}} +  \norm{\dot{u}_i(x_k) - \mathbf{\dot{u}}_{i,k}} \label{eq:localDE2} \end{equation}for
$i=0,1,\ldots,n-1$ and $k=0,1,\ldots,N_{i}N$. The next theorem gives a description of the local (space) discretization
error bound.
\begin{thm}
\label{thm:ADLerr-2}Suppose that the driver $f\in\mathcal{C}^{1,2}([0,T]\times\mathbb{R}^{2})$
and the terminal condition $g\in\mathcal{C}^{2}(\mathbb{R})$ and Assumptions \ref{assump:coef} and \ref{assump:tdisc} 
are satisfied, then the Fourier interpolation method yields a local
space discretization error of the form \begin{equation} \sup_{i,k}{ E}_{ik} = \mathcal{O}\left(\Delta x \right) + \mathcal{O}\left( e^{- K \Delta_i^{-s_0}  l^{q_0}} \right) \label{eq:ADerrbound-2} \end{equation}
for some constant $K>0$ on the alternative grid and under the trapezoidal
quadrature rule for any explicit $q$-stage Runge-Kutta scheme.\end{thm}
\begin{proof}
We follow the steps in the proof of Theorem \ref{thm:ADLerr}. The truncation error when computing the numerical solutions $\mathbf{\dot{u}}_{i,j,k}$
is\begin{align*} 
\lefteqn{ \CEsp{x_{ik} }{H^{\varphi_j}_{t_{i,j}, \gamma_j \Delta_i} \tilde{u}_{i+1}( t_{i+1}, X^{\pi}_{t_{i+1}} ; \beta_{j,1}) {\bf 1}_{\norm{\Delta X^{\pi}_i} > \frac{l}{2}} }{t_i} } \\
& < K \CEsp{x_{ik} }{ \left( H^{\varphi_j}_{t_{i,j}, \gamma_j \Delta_i} \right)^4 }{t_i}^{\frac{1}{4}}  \CEsp{x_{ik} }{ {\bf 1}_{\norm{\Delta X^{\pi}_i} > \frac{l}{2}} }{t_i} ^{\frac{1}{4}} \tag{using Cauchy-Schwarz inequality twice since $\tilde{u}_{i+1}( t_{i+1}, X^{\pi}_{t_{i+1}} ;.)$ is sq.\ int.} \\
& < K \Delta_i^{-\frac{1}{2}}\CEsp{x_{ik} }{  {\bf 1}_{\norm{\Delta X^{\pi}_i} > \frac{l}{2}} }{t_i}^{\frac{1}{4}} \tag{since $H^{\varphi_j}_{t_{i,j}, \gamma_j \Delta_i}$ is of Gaussian distribution} \\
& \leq K \Delta_i^{-\frac{1}{2}} \left( \inf_{s>0} e^{ -s \frac{l}{2} } \phi_i(- {\bf i}s )  + \inf_{s>0} e^{ -s \frac{l}{2} } \phi_i( {\bf i}s ) \right)^{\frac{1}{4}}
\tag{by Chernoff's inequality}\\
& < K \Delta_i^{-\frac{1}{2}}   e^{ -K_0 \Delta_i^{-s_0} l^{q_0}} \tag{by Assumptions \ref{assump:tdisc}}\\
& < K e^{-C \Delta_i^{-s_0} l^{q_0} }.
\end{align*}
The Fourier interpolation leads
to a first-order space discretization error when computing
the numerical solutions $\mathbf{\dot{u}}_{i,j,k}$ since the driver $f$ and the terminal condition $g$ are twice differentiable.

The same statements hold for the numerical solutions $\mathbf{u}_{i,2,k}$
using similar arguments. By recursion and using the Lipschitz property
of the driver $f$, the statements hold for $\mathbf{u}_{i,j,k}$,
$1<j\leq q+1$. Since the time step $t_{i}$ and the space node $x_{ik}$
are arbitrary, the space truncation and discretization error bounds
hold for any $i$ and $k$.
\end{proof}

Locally, the truncation error remains spectral. Nonetheless, it is of a unspecified index $q_0$ in this general setting where the conditional characteristic
function $\phi_{i}$ is itself unspecified. For higher order time discretizations, one can expect $q_0 \leq 2$ since the forward process increment $X^{\pi}_{t_{i+1}} - X^{\pi}_{t_{i}}$ has a heavy tail distribution. Indeed, the Gaussian distribution of forward process increments and the quadratic exponential form of their
characteristic functions were the main reason for the spectral convergence
of index $2$ of the truncation error in Section \ref{sec:ErrorA}. The space discretization error
though is unchanged with first-order due to the second-order differentiability
of the BSDE coefficients. However, the Fourier interpolation produces
a space discretization error with a higher order when the driver $f$
and the terminal function $g$ have the required smoothness.
In general, if $f\in\mathcal{C}_{b}^{m+1}$ and $g\in\mathcal{C}_{b}^{m+1}$,
we can expect a space discretization error of order $m$ which is
the convergence order of the underlying Fourier interpolation. 

We now turn to the global space discretization error defined as in equation (\ref{eq:globdiser}). The next theorem gives
its error bound.
\begin{thm}
\label{thm:GDR-2}Suppose the conditions of Theorem \ref{thm:ADLerr-2}
are satisfied. If the discretization is such that \begin{equation} \sup_{ i }  \left\{\frac{C_0 \Delta x}{{ \pi \Delta_i^{p_0}}}  \right\}  \leq 1    \label{eq:ccond2} \end{equation}then
the Fourier interpolation method is stable and yields a global discretization
error $E_{l,\Delta x}$ of the form \begin{equation} E_{l,\Delta x} = \mathcal{O}(\Delta x ) + \mathcal{O}\left(e^{-K \norm{\pi}^{-s_0} l^{q_0} }\right) \label{eq:GDR2} \end{equation}
where $K>0$ for any explicit $q$-stage Runge-Kutta scheme.\end{thm}
\begin{proof}
From the definition of the global space discretization error, we may
write \begin{align}
e_{ik} & \leq { E}_{n-i,k} + \norm{ \mathbf{{u}}_{n-i,k} - {u}_{n-i,k}}  \label{eq:rkproof1}\\
\dot{e}_{ik} & \leq {E}_{n-i,k} + \norm{ \mathbf{\dot{u}}_{n-i,k} - {\dot{u}}_{n-i,k}}. \label{eq:rkproof11}
\end{align}Assume that the boundary values of the function $\tilde{u}_{i+1}$
and the sequence $\tilde{u}_{i+1,s}$ are matched on the alternative
grid so that we don't have to treat the alternative transform. Under
an explicit $q-$stage Runge-Kutta scheme, we have
\begin{align}
  &\norm{ \mathbf{\dot{u}}_{i,j,k} - {\dot{u}}_{i,j,k}}
  = \norm{ \mathfrak{D}^{-1}\left[ \left\{ \Phi_{i,j}(\nu_{i+1,m}, x_{ik}) \mathbb{D}[ \tilde{u}_{i+1} - \tilde{u}_{i+1,s} ]_m     \right\}_{m=0}^{N_{i+1}N-1} \right]_{k+\frac{N}{2}} }   \nonumber \\
& \leq \frac{ \sum_{m=0}^{N_{i+1}N-1} \norm{\Phi_{i,j}(\nu_{i+1,m}, x_{ik} )} }{N_{i+1}N}  \sup_k \norm{ \tilde{u}_{i+1}(x_{ik},\beta_{1,j}) - \tilde{u}_{i+1,k} } \nonumber \\
& \leq \frac{\Delta x}{2 \pi}  \left(\int_{\mathbb{R}} \norm{\Phi_{i,j}(\nu,x_{i,k})} d \nu \right) \sup_k \norm{ \tilde{u}_{i+1}(x_{ik} ,\beta_{1,j}) - \tilde{u}_{i+1,k} } \nonumber \\
& \leq \frac{C_0 \Delta x}{2 \pi \Delta_i^{p_0}} \sup_k \norm{ \tilde{u}_{i+1}(x_{ik},\beta_{1,j}) - \tilde{u}_{i+1,k} } \tag{using Assumption \ref{assump:tdisc}} \nonumber \\
& \leq \frac{C_0 \Delta x}{2 \pi \Delta_i^{p_0}}  (1 + \Delta_i K) \sup_k e_{n-i-1,k} + \frac{C_0 \Delta x}{2 \pi \Delta_i^{p_0}}  \Delta_i K \sup_k \dot{e}_{n-i-1,k} \tag{since $f$ is Lipschitz and $\beta_{1,j}$ is bounded} \nonumber \\
& \leq \frac{C_0 \Delta x}{2 \pi \Delta_i^{p_0}}  (1 + \Delta_i K) \sup_k e_{n-i-1,k} + \frac{C_0 \Delta x}{2 \pi \Delta_i^{p_0}}  (1 + \Delta_i K) \sup_k \dot{e}_{n-i-1,k}.\label{eq:rkproofF1}
\end{align}
Similarly, we get
\begin{align*}
&  \norm{ \mathbf{{u}}_{i,2,k} - {{u}}_{i,2,k}}  \leq \norm{ \mathfrak{D}^{-1}\left[ \left\{ \phi_{i}(\nu_{i+1,m}, x_{ik}) \mathbb{D}[ \tilde{u}_{i+1} - \tilde{u}_{i+1,s} ]_m     \right\}_{m=0}^{N_{i+1}N-1} \right]_{k+\frac{N}{2}} }   \\
& \leq \frac{\Delta x}{2 \pi}  \left(\int_{\mathbb{R}} \norm{\phi_{i}(\nu,x_{i,k})} d \nu \right) \sup_k \norm{ \tilde{u}_{i+1}(x_{ik} ,\alpha_{1,2}) - \tilde{u}_{i+1,k} } \\
& \leq \frac{C_0 \Delta x}{2 \pi \Delta_i^{p_0}} \sup_k \norm{ \tilde{u}_{i+1}(x_{ik},\alpha_{1,2}) - \tilde{u}_{i+1,k} } \tag{using Assumption \ref{assump:tdisc}} \\
& \leq \frac{C_0 \Delta x}{2 \pi \Delta_i^{p_0}}  (1 + \Delta_i K) \sup_k e_{n-i-1,k} + \frac{C_0 \Delta x}{2 \pi \Delta_i^{p_0}}  (1 + \Delta_i K) \sup_k \dot{e}_{n-i-1,k}
\end{align*}
so that we get
\begin{align}
  \norm{ \mathbf{{u}}_{i,j,k} - {{u}}_{i,j,k}}
& \leq \frac{C_0 \Delta x}{2 \pi \Delta_i^{p_0}}  (1 + \Delta_i K)\left[\sup_k e_{n-i-1,k} + 2 \sup_k \dot{e}_{n-i-1,k} \right]
  \label{eq:rkproofF}
\end{align}
recursively for $1<j\leq q+1$ using the Lipschitz property of the driver $f$ and the boundedness of the Runge-Kutta coefficients.
Equations (\ref{eq:rkproof1}) and (\ref{eq:rkproof11}) combined with equations
(\ref{eq:rkproofF}) and (\ref{eq:rkproofF1}) lead to
\begin{align}
  \sup_k e_{i,k} + \sup_k \dot{e}_{i,k}
  & \leq 2 \sup_{i,k}E_{ik}  + \frac{C_0 \Delta x}{ \pi \Delta_i^{p_0}}  (1 + \Delta_{n-i} K)\left( \sup_k e_{i-1,k} + \sup_k \dot{e}_{i-1,k}  \right) \nonumber\\
  & \leq 2 \sup_{i,k}E_{ik} + \zeta  (1 + \Delta_{n-i} K)\left( \sup_k e_{i-1,k} + \sup_k \dot{e}_{i-1,k}  \right) \nonumber
\end{align}
where
\begin{equation*}
  \sup_i \left\{\frac{C_0 \Delta x}{{ \pi \Delta_i^{p_0}}}  \right\} \leq \zeta \leq 1.
\end{equation*}
Gronwall's Lemma then yields
\begin{equation}
  \sup_k e_{i,k} + \sup_k \dot{e}_{i,k} \leq 2 e^{TK} \sup_{i,k}E_{ik}  \label{eq:rkproof2}
\end{equation}
so that the scheme is stable. The result of equation (\ref{eq:GDR2})
follows by taking the supremum on the left hand side of equation (\ref{eq:rkproof2})
other time steps and applying Theorem \ref{thm:ADLerr-2}.
\end{proof}
In this general case, the global discretization error maintains the
structure of the local discretization error under a stability condition.
Equation (\ref{eq:ccond2}) indicates that the space discretization
has to be relatively as fine as the time discretization to ensure
stability. Hence, stability can always be reached for any time discretization
by refining the space discretization. However, the structure of the
characteristic functions $\phi_{i}$ and $\Phi_{ij}$ determines the
relative refinement needed for the space discretization.

\section{Numerical Results\label{sec:NumResults}}

We test the convergence properties of the Fourier interpolation method
on Runge-Kutta schemes with a problem of commodity derivative pricing
under a model proposed by \cite{luciaschwartz:2002}.
We shall test the method's convergence and behaviour on smooth and
unbounded FBSDE coefficients.

The commodity spot price $X$ is defined by \begin{equation} X_t = e^{ S(t) + V_t} \end{equation}
where the deterministic function $S:\mathbb{R}^{+}\rightarrow\mathbb{R}$
represents the seasonality component of the commodity and $V$ is
the price diffusion following an Ornstein-Uhlenbeck process according
to the \citet{vasicek:1977} model 
\begin{equation}  
dV_t = -\kappa V_tdt + \sigma dW_t.  
\end{equation}
As indicated by \cite{luciaschwartz:2002}, the
commodity spot price $X$ satisfies the stochastic differential equation
\begin{equation} dX_t = \kappa( \theta(t) - \ln X_t ) X_t dt + \sigma X_t dW_t \label{eq:comprice} \end{equation}
where \begin{equation} \theta(t) = \frac{1}{\kappa}\left( \frac{\sigma^2}{2} + \frac{dS}{dt}(t) \right) + S(t). \end{equation}
We consider the commodity price as our forward process through equation
(\ref{eq:comprice}).

When the risk-free rate $r$ and the market price of risk $\lambda$
are both constant, the forward (or future) price
$F_{t,T}:=Y_{t}=u(t,X_{t})$
with maturity $T>0$ at time $t<T$ is given by
\begin{align}
  Y_t &= \CEsp{ {\bf Q} }{X_T}{t} \nonumber \\
       &= \exp \left( S(T) +  (\ln X_t - S(t)) e^{-\kappa (T-t)} - \frac{\sigma \lambda}{\kappa}h(T-t,\kappa)  + \frac{\sigma^2}{4\kappa}h(T-t,2\kappa) \right) \label{eq:fwdp}
\end{align}
with \begin{equation}
  h(\tau,\kappa) = 1 - e^{-\kappa \tau}
\end{equation}
where the expectation is taken under the equivalent risk measure ${\bf Q}$.
It can be shown that the forward price solves a BSDE with linear driver
\begin{equation}
  f(t,x,y,z) = - \lambda z
\end{equation} and terminal
condition
\begin{equation}
  g(x) = x.
\end{equation}
Options on forward contracts can also be represented in form of BSDEs in this spot price
model but we limit our analysis to forward price estimation. From
equation (\ref{eq:fwdp}) the control process (or equivalently the
forward price delta) is given by
\begin{align}
  Z_t &= \sigma X_t \nabla u(t,X_t) \nonumber \\
  &= \sigma e^{-\kappa (T-t)} u(t,X_t).
\end{align}

The adjustment speed of the diffusion process is $\kappa=1.5$ and
the volatility of the diffusion is set to be $\sigma=0.065$. The
seasonality component is given by \begin{equation} S(t) = \ln {\bar P} + 0.05 \sin(2\pi t) \end{equation}
and the initial spot price by \begin{equation} X_0 = {\bar P} e^{V_0} = 0.95 {\bar P} \end{equation} 
where we normalize the real value%
\footnote{The real value $\bar{P}$ can be considered as the production cost
(per unit) of the commodity.%
} of the commodity $\bar{P}=1$. Also, the maturity of the forward
contract is $T=0.25$ and we suppose a market price of risk of $\lambda=0.25$.

The FBSDE is solved on an alternative grid centred at $X_{0}$ with
a uniform time mesh. For a given number of time steps $n$ and the
initial number $N_{0}=1$ of intervals, the length of an increment
interval is set as \begin{equation} l = \frac{1.8}{N_0 + n} \end{equation}
so that the truncated interval at time $t_{n}$ has length $1.8$.
This restriction keeps the space nodes in the upper half plane knowing
that the commodity price is a positive process. Moreover, the number
of space steps on an increment interval is $N=2$. 

We numerically solve the BSDE with the explicit $1-$stage Runge-Kutta
scheme of half-order and an explicit $2-$stage Runge-Kutta scheme
of first-order. Under the explicit $1-$stage scheme, the commodity
price is discretized with an Euler scheme whereas a Milstein scheme
is used for the forward process $X$ under the explicit $2-$stage
Runge-Kutta scheme. In
addition, we use an explicit $2-$stage Runge-Kutta scheme with tableau
\begin{center} 
\begin{tabular}{c| c c c| c c}
$0$    & $0$ & $0$    & $0$         & $0$ & $0$  \\
$\frac{2}{3}$ & $ \frac{2}{3} $ & $0$       & $0$ & $\frac{2}{3}$ & $0$    \\ \hline
$1$ & $ \frac{1}{4}$ & $\frac{3}{4}$ & $0$ & $1$ & $0$ \\
\end{tabular} 

\end{center} 

Under both FBSDE discretizations, we compute two different types of
error. The first error $E_{True}$ evaluates the maximal absolute
error of the numerical solution with respect to the true solution
\begin{align}
  E_{True} &= \max_{0 \leq i < n} \max_{ 0 \leq k \leq NN_i} \norm{ u(t_i, x_{ik}) - u_{ik} } 
  +  \max_{0 \leq i < n} \max_{ 0 \leq k \leq NN_i} \norm{ \dot{ u}(t_i, x_{ik}) - \dot{u}_{ik} }
\end{align}
where
\begin{equation}
  \dot{ u}(t,x) = \sigma x \nabla u(t,x) = \sigma e^{-\kappa (T-t)}u(t,x).
\end{equation}
The second error $E_{Sim}$ is a simulation error.
Given the numerical solution $\{X_{t_{i},j}^{\pi}\}_{j=1}^{m}$ ,
$i=0,1,\ldots,n-1$ with $m>0$ simulated paths for the forward process,
we compute the numerical solution $\{(y_{t_{i},j},z_{t_{i},j})\}_{j=1}^{m}$
of the backward processes by linearly interpolating the simulated paths through the BSDE numerical solutions $ \{ u_{ik}\}_{k=0}^{N_iN}$ and $ \{ \dot{u}_{ik}\}_{k=0}^{N_iN}$ at each time step $t_i$.
The error $E_{Sim}$ can be written as
\begin{align*}
  E_{Sim } &=  \frac{1}{m} \sum_{j=1}^{m} \left[ \max_{0 \leq i < n} \norm{ u(t_i, X^{\pi}_{t_i , j}) - y_{t_i,j} } 
  +  \left( \sum_{i = 0}^{n - 1} \Delta_i   ( \dot{u}(t_i, X^{\pi}_{t_i , j}) - z_{t_i,j})^2 \right)^{\frac{1}{2}}\right] .
\end{align*}
We systematically use $m=1000$ paths. Even if the errors $E_{True}$
and $E_{Sim}$ may be of the same order, they are interpreted differently.
The error $E_{True}$ gives the behaviour of the maximal approximation
error on the grid whereas $E_{Sim}$ gives the behaviour of the error
on the relevant part of grid when solving the FBSDE numerically. Figure
\ref{fig:1qRK} displays the errors under the explicit $1-$stage
Runge-Kutta scheme with $n\in\{5,10,20,50,100\}$ and Figure \ref{fig:2qRK}
shows the errors under the explicit $2-$stage scheme. 

\begin{figure}[p]
\begin{center}
\caption{Log-log plot of errors using the $1$-stage Runge-Kutta scheme.} 
\includegraphics[scale = 0.6]{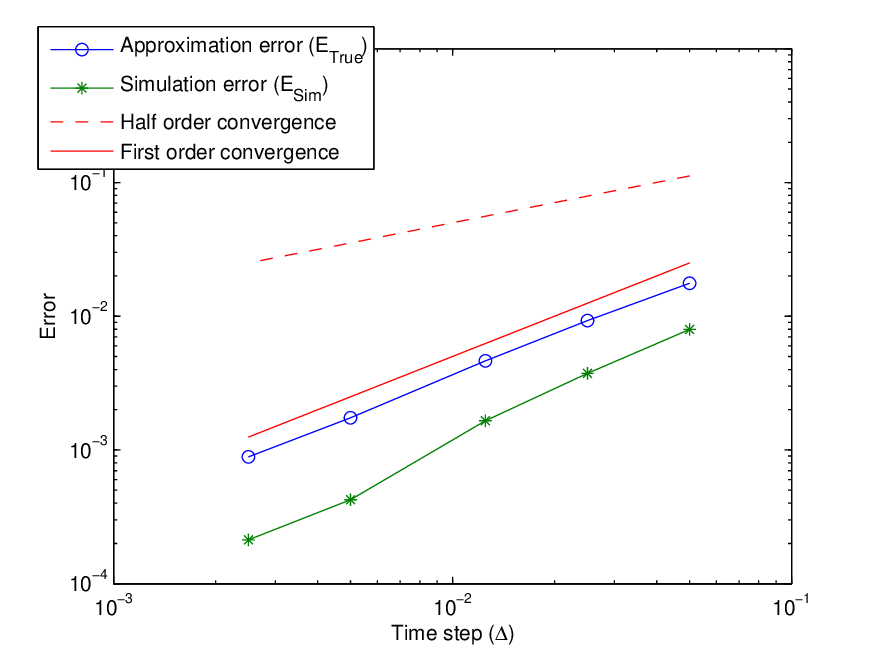} 
\label{fig:1qRK}

\footnotesize{The sample standard deviation of the error $E_{Sim}$ was less than $2\times 10^{-6}$ for all time discretizations.}

\end{center}
\end{figure}

\begin{figure}[p]
\begin{center}
\caption{Log-log plot of errors using the $2$-stage Runge-Kutta scheme.} 
\includegraphics[scale = 0.6]{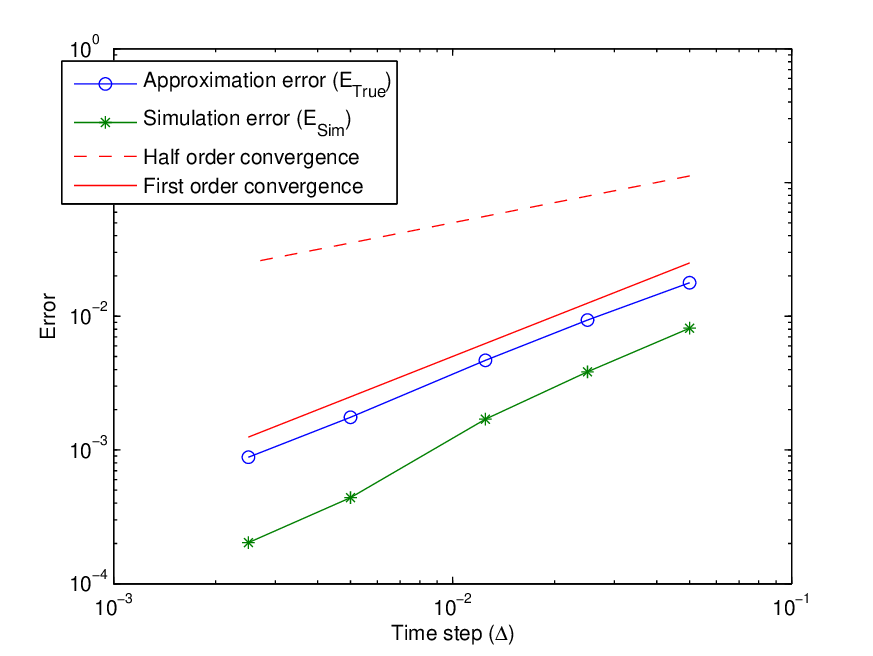} 
\label{fig:2qRK}

\footnotesize{The sample standard deviation of the error $E_{Sim}$ was less than $2\times 10^{-6}$ for all time discretizations.}

\end{center}
\end{figure}

The error graphs of Figures \ref{fig:1qRK} and \ref{fig:2qRK} look
almost identical and confirm that the $2-$stage scheme is of first
order and the $1-$stage scheme of (at least) half-order. The extra-efficiency
of the $1-$stage scheme may be attributed in this particular case
to the simplicity of the driver $f$ and the terminal condition $g$.
\begin{figure}[p]
\begin{center}
\caption{Simulation errors using the $2$-stage Runge-Kutta scheme.} 
\includegraphics[scale = 0.6]{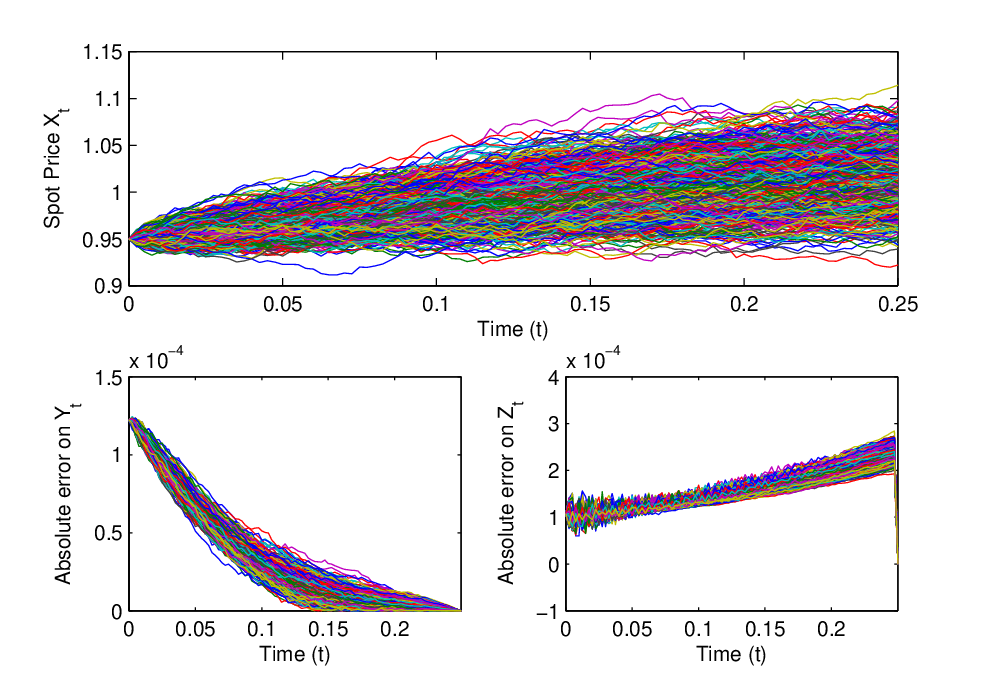} 
\label{fig:2qRKsim}

\footnotesize{The numerical solution is obtained on a time mesh with $n=100$ time steps and returns an forward price of $1.0121$ and initial value of $0.0453$ for the control process. The exact values are $1.0123$ and $0.0452$ respectively.}

\end{center}
\end{figure}\begin{figure}[p]
\begin{center}
\caption{Contour plot of errors using the $2$-stage Runge-Kutta scheme.} 
\includegraphics[scale = 0.6]{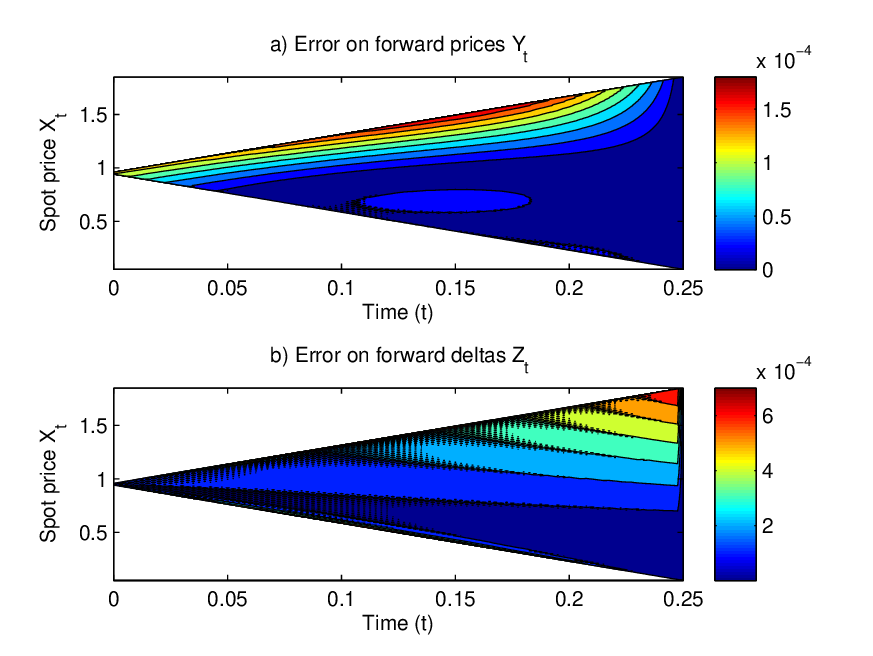} 
\label{fig:2qRKcontp}

\footnotesize{The numerical solution is obtained on a time mesh with $n=100$ time steps and returns an forward price of $1.0121$ and initial value of $0.0453$ for the control process. The exact values are $1.0123$ and $0.0452$ respectively.}

\end{center}
\end{figure}

In Figure \ref{fig:2qRKsim}, we present the absolute errors along
the simulated paths for the BSDE solution. One notices that the maximal
errors occur at the initial time $t_{0}=0$ for the forward price
($Y_{t}$) and at maturity $T=0.25$ for the control process ($Z_{t}$).
Nonetheless, the simulation errors are of the same order ($10^{-4}$)
for both processes. This information is confirmed by the contour plot
of Figure \ref{fig:2qRKcontp} not only along the simulated paths
but on the entire grid. 

Moreover, the contour plot gives indication on the source of errors.
Indeed, Figure \ref{fig:2qRKcontp} shows that the maximal errors
mainly occur for the upper space node values on the alternative grid
and they decrease for lower space node values. This is due to the
unbounded nature of the spot price process coefficients. Since the
volatility of the spot price is a positive and increasing function
of the spot price%
\footnote{See equation (\ref{eq:comprice}).%
}, higher spot price values lead to higher local volatility. Hence,
the fixed length of increment interval $l$ may not be sufficiently
large to ensure accuracy for higher space node values. In general,
the phenomenon is amplified with the magnitude of the forward process
coefficients as illustrated in the contour plot of Figure \ref{fig:2qRKconts}
where we choose a higher value for the volatility $\sigma$ and keep
the other parameters unchanged. Similar results can be obtained by
selecting a higher value for the speed of adjustment $\kappa$ as
shown in Figure \ref{fig:2qRKcontk} . \begin{figure}[p]
\begin{center}
\caption{Errors using the $2$-stage Runge-Kutta scheme with $\sigma = 0.08$.} 
\includegraphics[scale = 0.6]{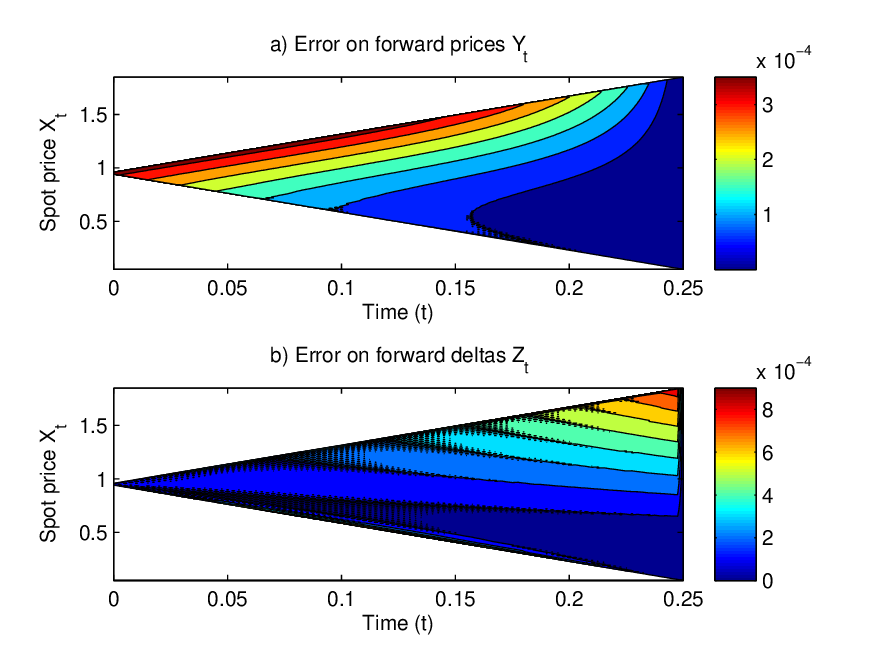} 
\label{fig:2qRKconts}

\footnotesize{The numerical solution is obtained on a time mesh with $n=100$ time steps and returns an forward price of $1.0115$ and initial value of $0.0558$ for the control process. The exact values are $1.0119$ and $0.0556$ respectively.}

\end{center}
\end{figure}\begin{figure}[p]
\begin{center}
\caption{Errors using the $2$-stage Runge-Kutta scheme with $\kappa = 3$.} 
\includegraphics[scale = 0.6]{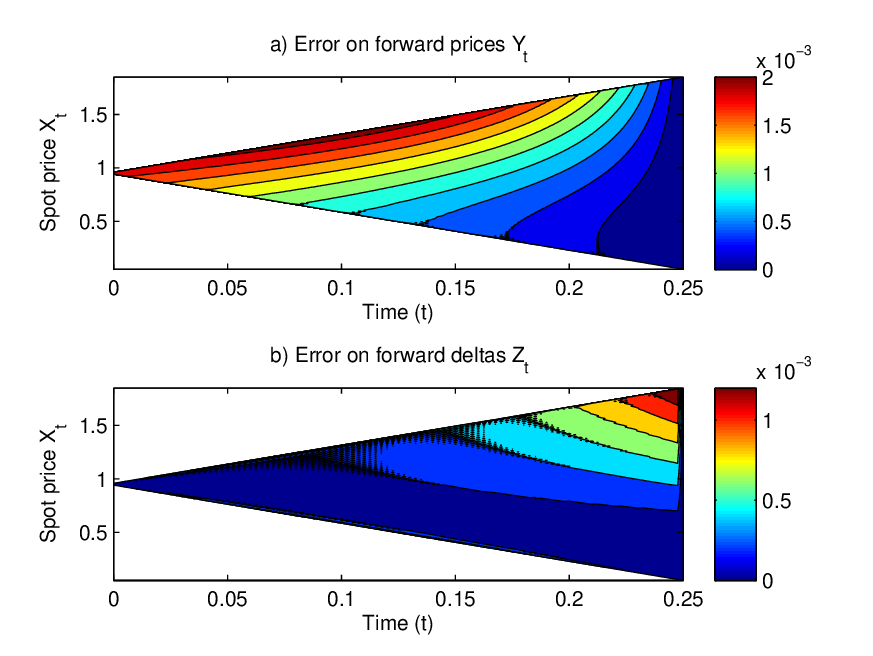} 
\label{fig:2qRKcontk}

\footnotesize{The numerical solution is obtained on a time mesh with $n=100$ time steps and returns an forward price of $1.0238$ and initial value of $0.0316$ for the control process. The exact values are $1.0257$ and $0.0315$ respectively.}

\end{center}
\end{figure}

We end this section with an efficiency study of our schemes. Using
the parameters initially given, the BSDE is solved on a uniform time
grid with $n\in\{10,20,40,50,60,80,100\}$ time steps and $N\in\{2,2^{2},2^{3},2^{4}\}$
space steps and value the computation time. Figure \ref{fig:cpu}
displays the results. First note that since the Fourier interpolation
method performs matrix multiplications,
it is much slower than the convolution method of \cite{hyndmanoyonongou:2013}.

As shown in Figure \ref{fig:cpu}, the computation time of Fourier
interpolation method increases with the number of time steps leading
to a trade-off between computation speed and accuracy. The exponential
nature of the curves suggests that preference has to be given to the
coarsest time discretization providing a satisfactory level of accuracy.
Similarly, the computation time also increase drastically with the
number $N$ of space steps. Coarse space grids insuring accuracy are
hence also preferable. Since a total number of $2q$ conditional expectations
are computed under a $q$-stage Runge-Kutta scheme, we can expect
the $1$-stage scheme to run twice as fast as the $2$-stage scheme.
This is confirmed on Figure \ref{fig:cpu}, especially when looking
at the computation times for $n=100$.

\begin{figure}
\begin{center}
\caption{CPU time (in seconds) of Runge-Kutta schemes.} 
\includegraphics[scale = 0.6]{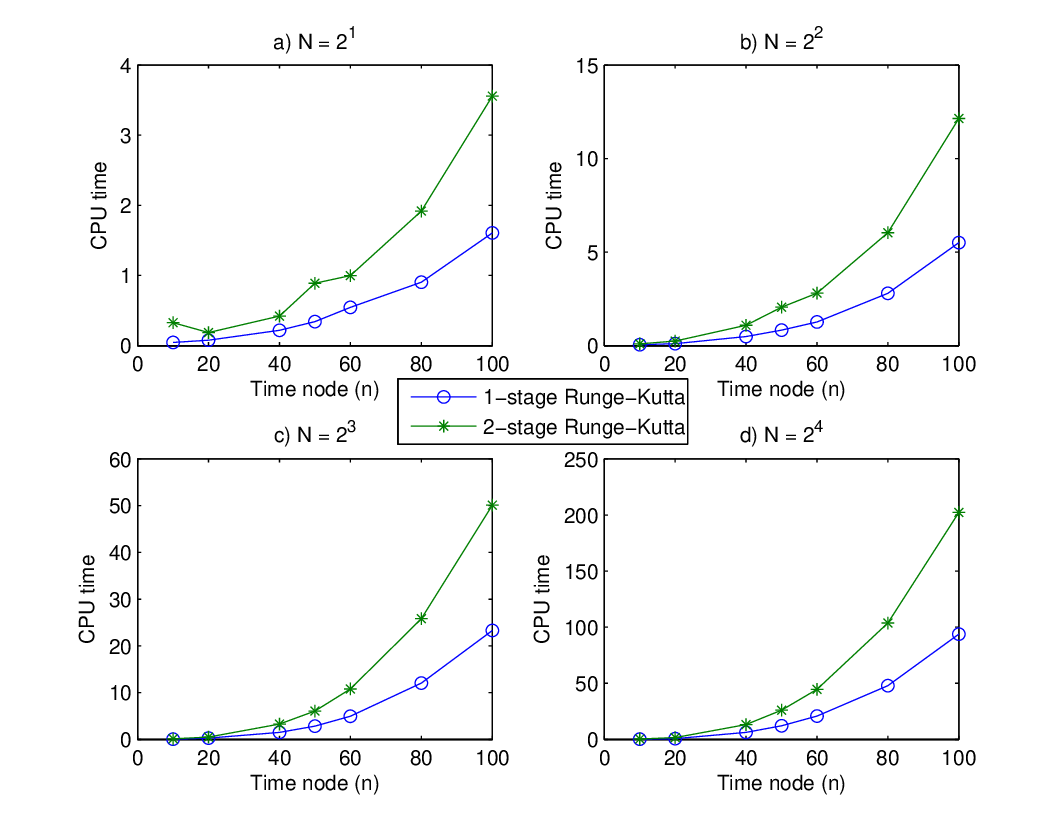} 
\label{fig:cpu}
\end{center}
\end{figure}

\section{Conclusion\label{sec:Conclusion}}

In order to solve the problem of extrapolation errors in the initial
implementation of the convolution method, we proposed an alternative
space discretization. The new tree-like space grid naturally allows
the usage of the FFT algorithm when computing the conditional expectation
included in the underlying explicit Euler scheme. The error analysis
shows that both the alternative grid and the (alternative) transform
suit the periodic nature of the FFT algorithm and help in producing
a stable, consistent and globally convergent numerical procedure for the FBSDE
approximate solutions. The second part of the paper deals with the implementation of the Fourier interpolation method with higher order time discretizations of FBSDEs. When the forward process increments admit conditional characteristic functions satisfying certain regularity conditions, it was shown that the method is also consistent, conditionally stable and globally convergent under Runge-Kutta schemes for FBSDEs.  A challenging area of research is the implementation of the methods of this paper in multidimensional and jump cases.

\vspace{5mm}
\noindent
\textbf{Acknowledgements}
\\ %
This research was supported by the Natural Sciences and Engineering Research Council of Canada (NSERC).

\bibliographystyle{abbrvnat}
\bibliography{conv-ii}

\begin{thebibliography}{36}
\providecommand{\natexlab}[1]{#1}
\providecommand{\url}[1]{\texttt{#1}}
\expandafter\ifx\csname urlstyle\endcsname\relax
  \providecommand{\doi}[1]{doi: #1}\else
  \providecommand{\doi}{doi: \begingroup \urlstyle{rm}\Url}\fi

\bibitem[Bally and Pages(2003)]{ballypages:2003}
V.~Bally and G.~Pages.
\newblock A quantization algorithm for solving multidimensional discrete-time
  optimal stopping problems.
\newblock \emph{Bernoulli}, 9\penalty0 (6):\penalty0 1003--1049, 2003.

\bibitem[Bally et~al.(2005)Bally, Pages, and Printems]{ballypagesp:2005}
V.~Bally, G.~Pages, and J.~Printems.
\newblock A quantization tree method for pricing and hedging multidimensional
  {A}merican options.
\newblock \emph{Math. Finance}, 15\penalty0 (1):\penalty0 119--168, 2005.

\bibitem[Beck et~al.(2019)Beck, {W}einan, and Jentzen]{beck}
C.~Beck, E.~{W}einan, and A.~Jentzen.
\newblock Machine learning approximation algorithms for high-dimensional fully
  nonlinear partial differential equations and second-order backward stochastic
  differential equations.
\newblock \emph{Journal of Nonlinear Science}, 29\penalty0 (4):\penalty0
  1563--1619, 2019.

\bibitem[Bender and Denk(2007)]{benderdenk:2007}
C.~Bender and R.~Denk.
\newblock A forward scheme for backward {SDE}s.
\newblock \emph{Stochastic Process. Appl.}, 117\penalty0 (12):\penalty0
  1793--1812, 2007.

\bibitem[Bender and Zhang(2008)]{benderzhang:2008}
C.~Bender and J.~Zhang.
\newblock Time discretization and {M}arkovian iteration for coupled {FBSDE}s.
\newblock \emph{Ann. Appl. Probab.}, 18\penalty0 (1):\penalty0 143--177, 2008.

\bibitem[Bouchard and Touzi(2004)]{bouchardtouzi:2004}
B.~Bouchard and N.~Touzi.
\newblock Discrete-time approximation and {M}onte-{C}arlo simulation of
  backward stochastic differential equations.
\newblock \emph{Stochastic Process. Appl.}, 111\penalty0 (2):\penalty0
  175--206, 2004.

\bibitem[Briand et~al.(2001)Briand, Delyon, and Memin]{briand:2001}
P.~Briand, B.~Delyon, and J.~Memin.
\newblock Donsker-type theorem for {BSDE}s.
\newblock \emph{Elect. Comm. in Probab.}, 6:\penalty0 1--14, 2001.

\bibitem[Chassagneux and Crisan(2014)]{Chasscrisan:2013}
J.~Chassagneux and D.~Crisan.
\newblock {R}unge-{K}utta schemes for {BSDE}s.
\newblock \emph{Ann. Appl. Probab.}, 24\penalty0 (2):\penalty0 679--720, 2014.

\bibitem[Chevance(1997)]{chevance:1997}
D.~Chevance.
\newblock Numerical methods for backward stochastic differential equations.
\newblock In L.~C.~G. Rogers and D.~Talay, editors, \emph{Numerical Methods in
  Finance}, Publ. Newton Inst., pages 232--244. Cambridge University Press,
  Cambridge, 1997.

\bibitem[Crisan and Manolarakis(2012)]{crisanm:2012}
D.~Crisan and K.~Manolarakis.
\newblock Solving backward stochastic differential equations using the cubature
  method: application to nonlinear pricing.
\newblock \emph{{SIAM} J. Financ. Math.}, 3\penalty0 (1):\penalty0 534--571,
  2012.

\bibitem[Crisan and Manolarakis(2014)]{crisanmol:2013}
D.~Crisan and K.~Manolarakis.
\newblock Second order discretization of backward {SDE}s and simulation with
  the cubature method.
\newblock \emph{Ann. Appl. Probab.}, 24\penalty0 (2):\penalty0 652--678, 2014.

\bibitem[Crisan et~al.(2010)Crisan, Manolarakis, and Touzi]{crisanmt:2010}
D.~Crisan, K.~Manolarakis, and N.~Touzi.
\newblock On the {M}onte {C}arlo simulation of {BSDE}s: An improvement on the
  {M}alliavin weights.
\newblock \emph{Stochastic Process. Appl.}, 120\penalty0 (7):\penalty0
  1133--1158, 2010.

\bibitem[Delarue and Menozzi(2006)]{delaruem:2006}
F.~Delarue and S.~Menozzi.
\newblock A forward-backward stochastic algorithm for quasi-linear {PDE}s.
\newblock \emph{Ann. Appl. Probab.}, 16\penalty0 (1):\penalty0 140--184, 2006.

\bibitem[{Douglas JR.} et~al.(1996){Douglas JR.}, Ma, and
  Protter]{douglas:1996}
J.~{Douglas JR.}, J.~Ma, and P.~Protter.
\newblock Numerical methods for forward-backward stochastic differential
  equations.
\newblock \emph{Ann. Appl. Probab.}, 6:\penalty0 940--968, 1996.

\bibitem[Duffie and Epstein(1992)]{duffieepstein:1992}
D.~Duffie and L.~G. Epstein.
\newblock Stochastic differential utility.
\newblock \emph{Econometrica}, 60\penalty0 (2):\penalty0 353--394, 1992.

\bibitem[{E}l {K}aroui et~al.(1997{\natexlab{a}}){E}l {K}aroui, Pardoux, and
  Quenez]{elkarouietal:1997}
N.~{E}l {K}aroui, E.~Pardoux, and M.~Quenez.
\newblock Reflected backward {SDE}s and {A}merican options.
\newblock In L.~C.~G. Rogers and D.~Talay, editors, \emph{Numerical Methods in
  Finance}, Publ. Newton Inst., pages 215--231. Cambridge University Press,
  Cambridge, 1997{\natexlab{a}}.

\bibitem[{E}l {K}aroui et~al.(1997{\natexlab{b}}){E}l {K}aroui, Peng, and
  Quenez]{elkaouri:1996}
N.~{E}l {K}aroui, S.~Peng, and M.-C. Quenez.
\newblock Backward stochastic differential equations in finance.
\newblock \emph{Math. Finance}, 7 (1):\penalty0 1--71, 1997{\natexlab{b}}.

\bibitem[Gobet et~al.(2005)Gobet, Lemor, and Warin]{gobet:2005}
E.~Gobet, J.-P. Lemor, and X.~Warin.
\newblock A regression-based {M}onte {C}arlo method to solve backward
  stochastic differential equations.
\newblock \emph{Ann. Appl. Probab.}, 15\penalty0 (3):\penalty0 2172--2202,
  2005.

\bibitem[Han and Long(2020)]{han20}
J.~Han and J.~Long.
\newblock Convergence of the deep {BSDE} method for coupled {FBSDE}s.
\newblock \emph{Probability, Uncertainty and Quantitative Risk}, 5\penalty0
  (1):\penalty0 5, 2020.

\bibitem[Huijskens et~al.(2016)Huijskens, Ruijter, and
  Oosterlee]{oosterlee:2016}
T.~Huijskens, M.~Ruijter, and C.~Oosterlee.
\newblock Efficient numerical {F}ourier methods for coupled forward-backward
  {SDEs}.
\newblock \emph{J. Comput. Appl. Math.}, 296:\penalty0 593--612, 2016.

\bibitem[Hyndman and {Oyono Ngou}(2017)]{hyndmanoyonongou:2013}
C.~B. Hyndman and P.~{Oyono Ngou}.
\newblock A convolution method for numerical solution of backward stochastic
  differential equations.
\newblock \emph{Methodol. Comput. Appl. Probab.}, 19:\penalty0 1--29, 2017.

\bibitem[Kloeden and Platen(1992)]{kloden:1992}
P.~Kloeden and E.~Platen.
\newblock \emph{Numerical Solution of Stochastic Differential Equations},
  volume~23 of \emph{Applications of Mathematics (New York)}.
\newblock Springer-Verlag, Berlin, 1992.

\bibitem[Lucia and Schwartz(2002)]{luciaschwartz:2002}
J.~Lucia and E.~Schwartz.
\newblock Electricity prices and power derivatives: Evidence from the nordic
  power exchange.
\newblock \emph{Rev. Derivatives Res.}, 5:\penalty0 5--50, 2002.

\bibitem[Ma et~al.(2002)Ma, Protter, Martin, and Torres]{maprottermt:2002}
J.~Ma, P.~Protter, J.~S. Martin, and S.~Torres.
\newblock Numerical method for backward stochastic differential equations.
\newblock \emph{Ann. Appl. Probab.}, 12 (1):\penalty0 302--316, 2002.

\bibitem[Ma et~al.(2008)Ma, Shen, and Zhao]{mashenzhao:2008}
J.~Ma, J.~Shen, and Y.~Zhao.
\newblock On numerical approximations of forward-backward stochastic
  differential equations.
\newblock \emph{{SIAM} J. Numer. Anal.}, 46\penalty0 (5):\penalty0 2636--2661,
  2008.

\bibitem[Oyono~Ngou(2014)]{poly:PhD}
P.~Oyono~Ngou.
\newblock \emph{Fourier methods for numerical solution of FBSDEs with
  applications in mathematical finance}.
\newblock PhD thesis, Concordia University, Montr\'eal, Canada, January 2014.

\bibitem[Pardoux and Peng(1992)]{pardouxpeng:1992}
E.~Pardoux and S.~Peng.
\newblock Backward stochastic differential equations and quasilinear parabolic
  partial differential equations.
\newblock In \emph{Stochastic partial differential equations and their
  applications ({C}harlotte, {NC}, 1991)}, volume 176 of \emph{Lec. Notes
  Control and Inform. Sci.}, pages 200--217. Springer, Berlin, 1992.

\bibitem[Peng and Xu(2011)]{pengxu:2011}
S.~Peng and M.~Xu.
\newblock Numerical algorithms for backward stochastic differential equations
  with 1-d {B}rownian motion: Convergence and simulations.
\newblock \emph{ESAIM Math. Model. Numer. Anal.}, 45:\penalty0 335--360, 2011.

\bibitem[Ruijter and Oosterlee(2015)]{RuijOosl:2013}
M.~Ruijter and C.~W. Oosterlee.
\newblock A {F}ourier-cosine method for an efficient computation of solutions
  to {BSDE}s.
\newblock \emph{{SIAM} J. Sci. Comput.}, 37\penalty0 (2):\penalty0 A859--A889,
  2015.

\bibitem[Ruijter and Oosterlee(2016)]{Ruijter20161}
M.~J. Ruijter and C.~W. Oosterlee.
\newblock Numerical {F}ourier method and second-order {T}aylor scheme for
  backward {SDE}s in finance.
\newblock \emph{Appl. Numer. Math.}, 103:\penalty0 1 -- 26, 2016.

\bibitem[Turkedjiev(2015)]{10.1214/EJP.v20-3022}
P.~Turkedjiev.
\newblock {Two algorithms for the discrete time approximation of Markovian
  backward stochastic differential equations under local conditions}.
\newblock \emph{Electronic Journal of Probability}, 20:\penalty0 1 -- 49, 2015.

\bibitem[Va{\v{s}}{\'{\i}}{\v{c}}ek(1977)]{vasicek:1977}
O.~Va{\v{s}}{\'{\i}}{\v{c}}ek.
\newblock An equilibrium characterization of the term structure.
\newblock \emph{J. Financ. Econ.}, 5\penalty0 (2), 1977.

\bibitem[{Weinan~E} et~al.(2017){Weinan~E}, Han, and Jentzen]{weinan17}
{Weinan~E}, J.~Han, and A.~Jentzen.
\newblock Deep learning-based numerical methods for high-dimensional parabolic
  partial differential equations and backward stochastic differential
  equations.
\newblock \emph{Communications in Mathematics and Statistics}, 5\penalty0
  (4):\penalty0 349--380, 2017.

\bibitem[{Weinan~E} et~al.(2019){Weinan~E}, Hutzenthaler, Jentzen, and
  Kruse]{weinan19}
{Weinan~E}, M.~Hutzenthaler, A.~Jentzen, and T.~Kruse.
\newblock On multilevel {P}icard numerical approximations for high-dimensional
  nonlinear parabolic partial differential equations and high-dimensional
  nonlinear backward stochastic differential equations.
\newblock \emph{Journal of Scientific Computing}, 79\penalty0 (3):\penalty0
  1534--1571, 2019.

\bibitem[{Weinan~E} et~al.(2022){Weinan~E}, Han, and Jentzen]{weinan20}
{Weinan~E}, J.~Han, and A.~Jentzen.
\newblock Algorithms for solving high dimensional {PDE}s: From nonlinear
  {M}onte {C}arlo to machine learning.
\newblock \emph{Nonlinearity}, pages 278--310, 2022.

\bibitem[Zhang(2004)]{zhang:2004}
J.~Zhang.
\newblock A numerical scheme for {BSDE}s.
\newblock \emph{Ann. Appl. Probab.}, 14:\penalty0 459--488, 2004.

\end{thebibliography}
\end{document}